\documentclass[11pt]{article}
\RequirePackage[OT1]{fontenc}
\RequirePackage{float}
\usepackage{amsmath}
\usepackage{amssymb}
\usepackage{bm}
\usepackage{bbm}
\usepackage{mathtools}
\usepackage{setspace}
\usepackage[margin=1in]{geometry}
\usepackage[utf8]{inputenc}
\usepackage[english]{babel}
\usepackage{mathrsfs}
\usepackage{subfigure}
\usepackage{amsthm}
\usepackage{afterpage}
\usepackage[authoryear]{natbib}
\usepackage{authblk}
\usepackage{multirow}
\usepackage{graphicx}
\usepackage{algorithm}
\usepackage{algpseudocode}
\usepackage{pifont}
\usepackage[colorlinks]{hyperref} % it is needed for todonotes package. 
\usepackage{booktabs}

\newtheorem{theorem}{Theorem}[section]

\newtheorem{lemma}[theorem]{Lemma}

\DeclareMathOperator*{\argmax}{arg\,max}

\newcommand{\bea}{\begin{eqnarray*}}
\newcommand{\eea}{\end{eqnarray*}}
\newcommand{\bean}{\begin{eqnarray}}
\newcommand{\eean}{\end{eqnarray}}

\newcommand{\lra}{\longrightarrow}

\newcommand{\sg}{\Sigma}
\newcommand{\what}{\widehat}

\newcommand{\bbR}{\mathbb{R}}
\newcommand{\bbE}{\mathbb{E}}

\begin{document}
\doublespacing
\title{\large \textbf{Bayesian inference on hierarchical nonlocal priors in generalized linear models} }
%	\footnote{Hi Kyoungjae, I made some edits, but did not sync promptly. Please refer to this version and make changes accordingly. Happy Holidays! :) }
\author[1]{Xuan Cao}
\author[2]{Kyoungjae Lee\footnote{Corresponding author.} }
\affil[1]{Department of Mathematical Sciences, University of Cincinnati}
\affil[2]{Department of Statistics, Sungkyunkwan University}

\maketitle
\begin{abstract}
Variable selection methods with nonlocal priors have been widely studied in linear regression models, and their theoretical and empirical performances have been reported.
%Variable selection methods using nonlocal priors have been thoroughly studied in linear regression. 
However, the crucial model selection properties for hierarchical nonlocal priors in high-dimensional generalized linear regression have rarely been investigated. 
In this paper, we consider a hierarchical nonlocal prior for high-dimensional logistic regression models and investigate theoretical properties of the posterior distribution.
Specifically, a product moment (pMOM) nonlocal prior is imposed over the regression coefficients with an Inverse-Gamma prior on the tuning parameter. 
%In this paper, we consider a logistic regression model with the product moment (pMOM) nonlocal prior over the regression coefficients and an appropriate Inverse-Gamma prior on the tuning parameter. 
%to analyze the underlying theoretical property. 
Under standard regularity assumptions, we establish strong model selection consistency in a high-dimensional setting, where the number of covariates is allowed to increase at a sub-exponential rate with the sample size. 
We implement the Laplace approximation for computing the posterior probabilities, and a modified shotgun stochastic search procedure is suggested for efficiently exploring the model space. 
We demonstrate the validity of the proposed method through simulation studies and an RNA-sequencing dataset for stratifying disease risk.

\end{abstract}


%%%%%%%%%%%%%%%%%%%%%%%%%%%%%%%%%%%%%%%%%%%%%
\section{Introduction} \label{sec:introduction}
The advance in modern technology has led to an increased ability to collect and store data on a large scale. 
This brings opportunities and, at the same time, tremendous challenges in analyzing data with a large number of covariates  per observation, the so-called high-dimensional problem.  
In high-dimensional analysis, variable selection is one of the very important tasks and commonly used techniques, especially in radiological and genetic research, with the high-dimensional data naturally extracted from imaging scans and gene sequencing.  
There is an extensive frequentist literature on variable selection, especially ones that are based on regularization techniques enforcing sparsity through penalization functions that share the common property of shrinkage toward sparse models \citep{Tibshirani:lasso,Fan:SCAD,zhang2010}.  
On the other hand, Bayesian model selection expresses the sparsity through a sparse prior, such as the popular spike and slab prior \citep{ishwaran2005spike,George1993,narisetty2014bayesian} and continuous shrinkage prior \citep{Liang:mixture:2008,johnson2012bayesian,liang2013bayesian}, i.e., a distribution that supports on the sparse model or model with sparse parameters, and inference is carried out through posterior inference. 

In this paper, we are interested in nonlocal priors \citep{John:Rossell:non-localtesting:2010} that are identically zero whenever a model parameter is equal to its null value. Compared to local priors, nonlocal prior distributions have relatively appealing properties for Bayesian model selection. Specifically, nonlocal priors discard spurious covariates faster as the sample size grows, while preserving exponential learning rates to detect nontrivial coefficients \citep{John:Rossell:non-localtesting:2010}. Under the setup of linear regression with $p$ predictors,
\citet{johnson2012bayesian} introduced the product moment (pMOM) nonlocal prior with density
%%%%%pMOM density
\begin{equation} 
\label{pmomdensity_introduction}
d_p(2\pi)^{-\frac p 2}(\tau \sigma^2)^{-rp - \frac p 2} |U_p| ^{\frac 1 2} \exp \Big(- \frac{ \beta_p ^ \top U_p  \beta_p}{2 
	\tau \sigma ^2}\Big)\prod_{i =1}^{p} \beta_{i}^{2r}. 
\end{equation}
\noindent
Here $U_p$ is a $p \times p$ nonsingular matrix, $r$ is a positive integer referred to as the order of the density and 
$d_p$ is the normalizing constant independent of the scale parameter $\tau$ and the variance $\sigma^2$. Variations of the density in 	
(\ref{pmomdensity_introduction}), called the piMOM and peMOM density, have also been developed in 
\cite{johnson2012bayesian} and \cite{RTJ:2013}. 
Under regularity conditions, \citet{johnson2012bayesian, Shin2018Sinica} and \citet{caoEntropy2020} demonstrated that the posterior distributions based on the pMOM and piMOM nonlocal prior densities can achieve strong model selection consistency in high-dimensional settings.
It implies that the posterior probability assigned to the true model converges to $1$ as the sample size grows. 
When the number of covariates is much smaller than the sample size, \cite{Shi:nonlocal} established the posterior convergence rate of the probability regarding the Hellinger distance between the posterior model and the true model under pMOM priors in a logistic regression model.

In the pMOM prior \eqref{pmomdensity_introduction}, the hyperparameter $\tau$ controls the dispersion of the density around the origin, and thus implicitly determines the magnitude of the regression coefficients that will be shrunk to zero \citep{johnson2012bayesian}. 
%the hyperparameter $\tau$ reflects the dispersion of the density around zero, thus implicitly determines the size of the regression coefficients that will be shrunk to zero \citep{johnson2012bayesian}. 
\citet{wu2020hyper} and \citet{Cao2020BA} extended the work in \cite{johnson2012bayesian} and \cite{Shin2018Sinica} by proposing a fully Bayesian approach with the pMOM nonlocal prior and an appropriate Inverse-Gamma prior on the hyperparameter $\tau$ referred to as the hyper-pMOM prior. 
In particular, \cite{wu2020hyper} investigated model selection properties of hyper-pMOM priors in generalized linear models (GLMs) under a fixed dimension $p$, and \citet{Cao2020BA} established strong model selection consistency of hyper-pMOM priors in linear regression when $p$ is allowed to grow at a polynomial rate of $n$. For the hyper-piMOM priors composed of the mixture of piMOM and Inverse-Gamma densities, \citet{BW:2017} established model selection consistency in generalized linear models under rather restrictive assumptions.

Despite recent developments in model selection using nonlocal priors, a rigorous Bayesian inference of hyper-pMOM priors in GLMs has not been undertaken to the best of our knowledge. 
%our first goal is to establish model selection consistency as well as posterior contraction rate for hierarchical nonlocal priors, particularly, the hyper-pMOM priors in a high-dimensional generalized linear model. It is known that the computation problem can arise for Bayesian approaches due to the non-conjugate nature of priors in generalized linear regression. Our second goal is to develop efficient algorithms for exploring the massive posterior space and to conduct real data analysis using the proposed method. These are challenging goals, as the posterior distributions are not available in closed form for nonlocal priors and the extra prior layer of prior creates more theoretical and computational challenges for high-dimensional posterior analysis. 
Motivated by this gap, we establish model selection consistency of the hyper-pMOM prior on regression coefficients in a GLM, in particular, logistic regression when the number of covariates grows at a sub-exponential rate of the sample size (Theorems \ref{thm:nosuper} to \ref{thm:selection}). 
%in a GLM, in particular, logistic regression, with hyper-pMOM prior on regression coefficients (Theorems \ref{thm:nosuper} to \ref{thm:selection}) when the number of covariates grows at a sub-exponential rate of the sample size. 
Furthermore, it is known that the computation problem can arise for Bayesian approaches due to the non-conjugate nature of priors in GLMs. 
To address this issue, we obtain posterior probabilities via Laplace approximation and then implement a slightly modified shotgun stochastic search algorithm for exploring the sparsity pattern of the regression coefficients. 
We demonstrate that the proposed method can outperform existing state-of-the-art methods including both penalized likelihood and Bayesian approaches in various settings. 
Finally, the proposed method is applied to an RNA-sequencing dataset consisting of gene expression levels to identify differentially expressed genes for disease risk stratification.

The rest of paper is organized as follows. Section \ref{sec:model spec} provides background material regarding GLMs and revisits the hyper-pMOM distribution.
We detail strong selection consistency results in Section \ref{sec:selection cons}, and proofs are provided in the supplement. The posterior computation algorithm is described in Section \ref{sec:computation}, and we show the performance of the proposed method and compare it with other competitors through simulation studies in Section \ref{sec:sim}. In Section \ref{sec:real}, we conduct a data analysis for predicting asthma and show that the hyper-pMOM prior yields better prediction performance compared with other contenders. We conclude with a discussion in Section \ref{sec:disc}.

\section{Methodology} \label{sec:model spec}
\subsection{Variable Selection in Logistic Regression}
We first describe the framework and introduce some notations for Bayesian variable selection in logistic regression. 
Let $ y \in \{0,1\}^{n}$ be the binary response vector and $ X \in \bbR^{n\times p}$ be the design matrix. Without loss of generality, we assume that the columns of $X$ are standardized to have zero mean and unit variance. 
Let $ x_i \in \bbR^p$ denote the $i$th row vector of $X$ that contains the covariates for the $i$th subject. Let $ \beta$ be the $p \times 1$ vector of regression coefficients.
We first consider the standard logistic regression model:
\begin{align} \label{logistic_model}
P \big(y_i = 1 \mid  x_i,  \beta \big) = \frac{\exp\big( x_i^\top \beta\big)}{1+\exp\big( x_i^\top \beta\big)}, \quad \text{ for } i = 1,2, \ldots, n .
\end{align}
We present a scenario where the dimension of predictors, $p$, grows with the sample size $n$. 
Thus, the number of predictors is a function of $n$, that is, $p = p_n$, but we denote it as $p$ for notational simplicity.
The goal of this paper is variable selection, i.e., to correctly identify all the locations of nonzero regression coefficients. 

We denote a model by $ k = \left\{k_1, k_2, \ldots, k_{| k|}\right\} \subseteq [p] =: \{1,2, \ldots, p \}$ if and only if all the nonzero elements of $ \beta$ are $\beta_{k_1}, \beta_{k_2}, \ldots, \beta_{k_{| k|}}$, where $| k|$ is the cardinality of $ k$. 
For any $\beta \in \bbR^p$ and $k \subseteq [p]$,  let $ \beta_{ k} = \big(\beta_{k_1}, \beta_{k_2}, \ldots, \beta_{k_{| k|}}\big)^\top \in \bbR^{|k|}$. 
Similarly, for any $m \times p$ matrix $A$ and $k \subseteq [p]$, let $ {A}_{ k} \in  \bbR^{m\times | k|}$ denote the submatrix of $A$ containing the columns of $ A$ indexed by model $ k$. In particular, for any $1 \le i \le n$ and $k \subseteq [p]$, we denote $ x_{i k} \in \bbR^{|k|}$ as the subvector of $ x_{i} \in \bbR^p$ containing the entries of $ x_{i}$ corresponding to model $ k$.

\subsection{Hierarchical Nonlocal Priors}
The class of the following hierarchical nonlocal priors can be used for variable selection:
\bean
\pi\left( {\beta_k} \mid \tau,  k\right) &=& d_{ k}(2\pi)^{-\frac {| k|} 2}(\tau)^{-r| k| - \frac {| k|} 2} | {U_k}| ^{\frac 1 2} \exp \Big(- \frac{ {\beta_k}^\top  {U_k} {\beta_k}}{2 \tau}\Big)\prod_{i =1}^{| k|} \beta_{k_i}^{2r}, \quad \beta_k \in \bbR^{|k|} , \label{model:pmom}\\
\pi(\tau) &=& \frac{\lambda_2^{\lambda_1}}{\Gamma(\lambda_1)}\tau^{-\lambda_1 - 1} \exp\Big(-\frac{\lambda_2}\tau\Big), \quad \tau >0 , \label{model:tau}
\eean
where $U$ is a $p \times p$ nonsingular matrix, $r$ is a positive integer and $\lambda_1, \lambda_2$ are positive constants.
we refer to the mixture of densities of pMOM and Inverse-Gamma in (\ref{model:pmom}) and (\ref{model:tau}) as the {hyper-pMOM prior} \citep{wu2020hyper,Cao2020BA}. 
It is easy to see that the marginal density of $ \beta_k$, after integrating out $\tau$, has the following form:
\bea
\pi\left( \beta_k \mid k\right)
&=& \int \pi(\beta_k \mid \tau, k ) \pi(\tau)  d \tau \\
&=&\frac{\lambda_2^{\lambda_1}}{\Gamma(\lambda_1)} \frac{\Gamma(r|k|+\frac {|k|} 2 +\lambda_1)}{(\lambda_2 + \frac{ \beta_k ^ \top U_k  \beta_k}{2})^{r|k| + \frac {|k|} 2 + \frac 1 2}}d_k(2\pi)^{-\frac {|k|} 2}|U_k| ^{\frac 1 2} \prod_{i =1}^{|k|} \beta_{k_i}^{2r}.
\eea
Compared to the pMOM density in \eqref{model:pmom} with given $\tau$, $\pi\left( \beta_k \mid k\right)$ possesses thicker tails and could achieve better model selection performance especially for small samples.  
For more details, see \cite{Liang:mixture:2008}, for example, that investigates the finite sample performance of hyper-$g$ priors.

For the prior over the model space, we suggest using the following uniform prior and restricting the analysis to models with a size of less than or equal to $m_n$:
\bean
\pi(k) &\propto& \mathbbm{1}(|k| \le m_n).  \label{uniform_prior}
\eean
Similar structure has also been considered in \cite{Naveen:2018,Shin2018Sinica} and \cite{Cao2020BA}. 
As an alternative to the uniform prior \eqref{uniform_prior}, one may also consider the complexity prior \citep{castillo2015bayesian}. However, as noted in \cite{Shin2018Sinica}, the penalty over large models can be derived directly from the nonlocal densities themselves without the extra penalization through the prior over the model space. 
In particular,  \cite{Cao2020BA} conducted simulation studies to compare the model selection results under a uniform prior and a complexity prior, and they showed the superior performance of model selection under a uniform prior.
%In particular,  \cite{Cao2020BA} conducted the simulation studies to compare the model selection results under uniform prior and complexity prior, and showed the superior performance of model selection under uniform prior compared with that under the complexity prior. 

Note that in the hierarchical nonlocal prior \eqref{logistic_model} to \eqref{uniform_prior}, no specific conditions have yet been assigned to the hyperparameters. 
Some standard regularity assumptions on the hyperparameters will be provided in Section \ref{sec:selection cons}. 

By the hierarchical model \eqref{logistic_model} to \eqref{uniform_prior} and Bayes' rule, the resulting posterior probability for model $ k$ is denoted by
\begin{align*}
\pi( k \mid  y) = \frac{\pi( k)}{m( y)}m_{ k}( y),
\end{align*}
where $m( y)$ is the marginal density of $y$, and $m_{ k}( y)$ is the marginal density of $ y$ under 
model $ k$ given by
	\bean \label{marginal}
	\begin{split}
		m_{ k}( y) &= \iint \exp \big\{L_n( {\beta_{k}}) \big\} \pi\left( {\beta_{k}} \mid  \tau, k\right) \pi(\tau) \, d\tau  d {\beta_{k}}   \\
		&=   \int \exp \big\{L_n( {\beta_{k}}) \big\} \pi\left( {\beta_{k}} \mid   k\right) \,d {\beta_{k}} , 
	\end{split}
\eean 
where 
\bean \label{loglikehood_logistic}
L_n( {\beta_{k}}) = \log\Bigg(\prod_{i =1}^n\bigg\{\frac{\exp\big( x_{i k}^\top {\beta_k}\big)}{1+\exp\big( x_{i k}^\top {\beta_k}\big)}\bigg\}^{y_i}\bigg\{\frac{1}{1+\exp\big( x_{i k}^\top {\beta_k}\big)}\bigg\}^{1-y_i}\Bigg)
\eean 
is the log-likelihood function. 
The above marginal posterior probabilities for model $k$ can be used to find the posterior mode,
\begin{equation} \label{a4}
\hat{ k} =  \argmax_{ k} \pi({ k} \mid  y).
\end{equation} 
The closed form of these posterior probabilities cannot be obtained due to the non-conjugate nature of nonlocal densities. Therefore, special efforts need to be devoted for both consistency results and computational strategy as we shall see in the following sections. 
In Section \ref{sec:computation}, we will adopt a (modified) stochastic search algorithm that utilizes posterior probabilities to target the mode in a more efficient way compared with Markov chain Monte Carlo (MCMC).

\subsection{Extension to Generalized Linear Model}
In this section, we extend our previous discussion on logistic regression to a GLM. Given predictors $ x_i$ and an outcome $y_i$ for $1 \le i \le n$, a GLM has a probability density function or probability mass function of the form
$$p(y_i\mid\theta_i ) = \exp\big\{a(\theta_i)y_i + b(\theta_i) + c(y_i)\big\},$$
in which $a(\cdot)$ is a continuously differentiable function with respect to $\theta$ with nonzero derivative, $b(\cdot)$ is also a continuously differentiable function of $\theta$, $c(\cdot)$ is some constant function of $y$, and $\theta_i = \theta_i(\beta) =  x_i^\top  \beta$ is the natural parameter.
%that relates the response to the predictors through the linear function $\theta =  x_i^\top  \beta$. 
%The mean function is $\mu = E(y_i |  x_i) = -b^\prime(\theta)/a^\prime(\theta) \triangleq \phi(\theta)$, where $\phi(\cdot)$ is the inverse of some chosen link function. 

The class of hierarchical pMOM densities specified in \eqref{model:pmom} and \eqref{model:tau} can still be used for model selection in the generalized setting by noting that the log-likelihood function in \eqref{marginal} and \eqref{loglikehood_logistic} now takes the general form of 
\bea
L_n({\beta_k}) = \sum_{i = 1}^n\big\{a(\theta_i({\beta_k}))y_i + b(\theta_i({\beta_k})) + c(y_i)\big\}.
\eea
Using similar techniques in Section \ref{sec:computation}, one can also develop efficient search algorithms based on different log-likelihood functions to navigate the posterior mode through the model space.

\section{Main Results}\label{sec:selection cons}

In this section, we show that the hyper-pMOM prior enjoys desirable model selection properties in a GLM.
Let $ t = \{ t_1, t_2,\ldots, t_{|t|}\} \subseteq [p]$ be the true model, which means that the nonzero locations of the true coefficient vector are $ t = (j , j \in  t)$. We consider $| t|$ to be a fixed value. Let $ \beta_0 \in \bbR^p$ be the true coefficient vector and $ \beta_{0,  t} \in \bbR^{| t|}$ be the vector of the true nonzero coefficients.
For a given model $ k \subseteq [p]$, we denote ${L_n( {\beta_{k}})}$ and $ s_n( {\beta_{k}}) = \partial L_n( {\beta_{k}}) /(\partial  {\beta_{k}})$ as the log-likelihood and score function, respectively. 
In the following analysis, we will focus on logistic regression, but our argument can be extended to any other GLMs such as a probit regression model by imposing certain conditions on the design matrix to effectively bound the Hessian matrix.
Let
\bea
 H_n(  {\beta_{k}}) &=& - \frac{\partial^2 L_n( {\beta_{k}}) }{\partial  {\beta_{k}} \partial  {\beta_{k}}^\top} \,\,=\,\, \sum_{i=1}^n \sigma_i^2( {\beta_{k}})  x_{ik}  x_{ik}^\top \,\,=\,\,  {X_{k}}^\top  \sg( {\beta_{k}})  {X_{k}}
\eea
be the negative Hessian of $L_n( {\beta_{k}})$, where $ \sg( {\beta_{k}}) \equiv  \sg_{k} = {\rm diag}(\sigma_1^2( {\beta_{k}}),\ldots, \sigma_n^2( {\beta_{k}}))$, $\sigma_i^2( {\beta_{k}}) = \mu_i( {\beta_{k}})(1- \mu_i( {\beta_{k}}))$ and 
\bea
\mu_i( {\beta_{k}}) &=& \frac{\exp \big( x_{i  k}^\top  {\beta_{k}} \big) }{1+ \exp \big( x_{i  k}^\top  {\beta_{k}} \big)}.
\eea
In the rest of the paper, we denote $ \sg =  \sg ( {\beta_{0,t}})$ and $\sigma_i^2 = \sigma_i^2( {\beta_{0,t}})$ for simplicity.

Before establishing our main results, we introduce the following notation. 
For any $a, b \in \bbR$, $a\vee b$ and $a \wedge b$ mean the maximum and minimum of $a$ and $b$, respectively.
For any positive real sequences $a_n$ and $b_n$, we denote $a_n \lesssim b_n$, or equivalently $a_n =O(b_n)$, if there exists a constant $C>0$ such that $|a_n| \le C |b_n|$ for all large $n$.
We denote $a_n \ll b_n$, or equivalently $a_n =o(b_n)$, if $a_n / b_n \lra 0$ as $n\to\infty$.
We denote  $a_n \sim b_n$, if there exist constants $C_1>C_2>0$ such that $C_2 < b_n/ a_n \le a_n / b_n < C_1$.
The $\ell_2$-norm for a given vector $ v = (v_1,v_2,\ldots, v_p)^\top \in \bbR^p$ is defined as $\| v\|_2 = ( \sum_{j=1}^p v_j^2 )^{1/2}$.
For any real symmetric matrix $ A$, let $\lambda_{\max}( A)$ and $\lambda_{\min}( A)$ be maximum and minimum eigenvalue of $ A$, respectively.
We assume the following standard conditions for obtaining the asymptotic results: \\
\noindent{\bf Condition (A1)}\label{cond_A1} $\log n \lesssim \log p =o(n)$ and $m_n = O\Big((n /\log p)^{\frac{1-d'}{2}}  \wedge p\Big)$, as $n\to \infty$.\\
\noindent{\bf Condition (A2)}\label{cond_A2} For some constant $C >0$ and $0 \le d < (1+d)/2 \le d' \le 1$,
\bea
\max_{i,j} |x_{ij}| &\le& C , \\
0 < \lambda  \le \min_{ k : | k| \le m_n} \lambda_{\min}\Big( n^{-1}H_n( \beta_{0,  k}) \Big) 
%&\le& \max_{k : | k| \le m_n + | t|} \lambda_{\max}\Big( n^{-1}  X_{k}^\top X_{k}  \Big)  
&\le& \Lambda_{m_n} \le C^2 \Big( \frac{n}{\log p} \wedge \log p\Big)^d,
\eea 
%and $| t| \le   m_n$, where  $m_n = \big\{  (  n /\log p)^{\frac{1-d'}{2}}  \wedge p\big\}$
 and $\Lambda_{\zeta} = \max_{ k : | k| \le \zeta} \lambda_{\max} ( n^{-1}   X_{ k}^\top  X_{ k})$ for any integer $\zeta >0$. 
Furthermore, $\| \beta_{0, t}\|_2^2 = O\big((\log p)^d \big)$.\\
\noindent{\bf Condition (A3)}\label{cond_A3} 
For some  constant $c_0>0$,
%\footnote{KL: This condition seems too restrictive. For example, it means that the minimum nonzero entry in $\beta_0$ is larger than $|t| p /\log n$ which diverges to infinity as $n\to\infty$. Can this condition be weakened? Or, did you mean $\frac{\log p}{n}$ instead of $\frac{p}{\log n}$? XC: Sorry Kyoungjae. I meant $\frac{\log p}{n}$.}
\bea
\min_{j \in  t}  \|\beta_{0, j}\|_2^2 &\ge&  c_0 | t| \Lambda_{| t|} \Big(\frac{\log p}{n} \vee \frac 1{\log p}\Big).
\eea
\noindent{\bf Condition (A4)}\label{cond_A4} For some constants $ \delta, a_1, a_2 > 0$, the hyperparameters satisfy 
\bea
 a_1 < \lambda_1 < a_2, \,\,\, \lambda_2^{r + 1/2} \sim  n^{-1/2} p^{2+\delta} \,\mbox{ and }\,  a_1 < \lambda_{\min}(U) \le \lambda_{\max}(U) < a_2.
\eea
Condition \hyperref[cond_A1]{\rm (A1)} ensures our proposed method can accommodate high dimensions where the number of predictors grows at a sub-exponential rate of $n$.
Condition \hyperref[cond_A1]{\rm (A1)} also specifies the parameter $m_n$ in the uniform prior \eqref{uniform_prior} that restricts our analysis on a set of {\it reasonably large} models.
Similar assumptions restricting the model size have been commonly assumed in the sparse estimation literature \citep{liang2013bayesian,Naveen:2018,Shin2018Sinica,LLL:2019}.

Condition \hyperref[cond_A2]{\rm (A2)} gives  lower and upper bounds of $\lambda_{\min} \big(n^{-1} H_n( \beta_{0, k} ) \big)$ and $\lambda_{\max} \big(n^{-1}  X_{ k}^\top  X_{ k} \big)$, respectively, where $k$ belongs to the set of reasonably large models.
The lower bound condition can be seen as a restricted eigenvalue condition for $k$-sparse vectors and is satisfied with high probability for sub-Gaussian design matrices \citep{Naveen:2018}.
Similar conditions have been used in the linear regression literature \citep{ishwaran2005spike,yang2016computational,song2017nearly}. 

Condition \hyperref[cond_A2]{\rm (A2)} also allows the magnitude of true signals to increase to infinity but stay bounded above by $(\log p)^d$ up to some constant, while Condition \hyperref[cond_A3]{\rm (A3)}, the well-known {\it beta-min} condition, gives a lower bound for nonzero signals. In general, this type of condition is necessary to not neglect any small signals. 

Condition \hyperref[cond_A4]{\rm (A4)} suggests appropriate conditions for the hyperparameters in \eqref{model:pmom} and \eqref{model:tau}. Similar assumption has also been considered in \cite{Shin2018Sinica}, \cite{johnson2012bayesian} and \cite{Cao2020BA}. In particular, we extend the previous polynomial rate of the dimension in \cite{Cao2020BA} by considering a larger order of the hyperparameter $\lambda_2$.

\subsection{Model Selection Consistency}

\begin{theorem}[No super set]\label{thm:nosuper}
	Under Conditions \hyperref[cond_A1]{\rm (A1)}, \hyperref[cond_A2]{\rm (A2)} and \hyperref[cond_A4]{\rm (A4)}, 
	\bea
	\pi  \big( k \supsetneq  t \mid  y\big) &\overset{P}{\lra}& 0  , \,\, \text{ as } n\to\infty.
	\eea
\end{theorem}

Theorem \ref{thm:nosuper} says that, asymptotically, our posterior does not overfit the model, i.e., it does not include unnecessarily many variables.
Of course, the result does not guarantee that the posterior will concentrate on the true model.
To capture every significant variable, we require the magnitudes of nonzero entries in $ \beta_{0,  t}$ not to be too small.
Theorem \ref{thm:ratio} shows that with an appropriate lower bound specified in Condition \hyperref[cond_A3]{\rm (A3)}, the true model $t$ will be the mode of the posterior.

\begin{theorem}[Posterior ratio consistency]\label{thm:ratio}
	Under Conditions \hyperref[cond_A1]{\rm (A1)}--\hyperref[cond_A4]{\rm (A4)} with $c_0 = \{(1-\epsilon_0)\lambda\}^{-1} \big[ 2(3+\delta) + 5 \{(1-\epsilon_0)\lambda\}^{-1}  \big]$ for some small constant $\epsilon_0>0$,
	\bea
	\max_{ k \neq  t} \frac{\pi  \big( k \mid  y \big)}{\pi  \big( t \mid  y\big)} &\overset{P}{\lra}& 0  , \,\, \text{ as } n\to\infty .
	\eea
\end{theorem}

Posterior ratio consistency is a useful property especially when we are interested in the point estimation with the posterior mode, but does not provide how large is the probability that the posterior puts on the true model.
In the following theorem, we state that our posterior achieves {\it strong selection consistency}.
By strong selection consistency, we mean that the posterior probability assigned to the true model $t$ converges to 1, which requires a slightly stronger condition on the lower bound for the magnitudes of nonzero entries in $ \beta_{0,  t}$ compared to that in Theorem \ref{thm:ratio}.
\begin{theorem}[Strong selection consistency]\label{thm:selection}
	Under Conditions \hyperref[cond_A1]{\rm (A1)}--\hyperref[cond_A4]{\rm (A4)} with $c_0 = \{(1-\epsilon_0)\lambda\}^{-1}  \big[ 2(9 + 2\delta) + 5\{(1-\epsilon_0)\lambda\}^{-1}  \big]$ for some small constant $\epsilon_0>0$, the following holds:
	\bea
	\pi  \big( t \mid  y\big) &\overset{P}{\lra}& 1  , \,\, \text{ as } n\to\infty  .
	\eea
\end{theorem}

%\subsection{Posterior Contraction Rate}

\subsection{Comparison with Existing Work}

We compare our results and assumptions with those of existing methods using nonlocal priors in generalized linear regression. \citet{Shi:nonlocal} established the posterior convergence rate for nonlocal priors under the assumption of $p\log(1/\epsilon_n^2) \ll n\epsilon_n^2$ for some $\epsilon_n \in (0,1]$ satisfying $n\epsilon_n^2 \gg 1$, which indicates that $p$ can increase with the sample size but slower than $n$. \citet{wu2020hyper} investigated the model selection performance of hyper-nonlocal priors that combine the Fisher information matrix with the pMOM density and established asymptotic properties under a fixed dimension of predictors. Both works considered the setting of low to moderate dimensions, while we allow $p$ to grow at a sub-exponential rate of $n$, the so-called ``ultra high-dimensional" setting \citep{Shin2018Sinica}.

\cite{BW:2017} considered the following hyper-piMOM priors for regression coefficients in GLMs and established the high-dimensional model selection consistency:
\begin{eqnarray*}
	\beta_k \mid \tilde{\tau}_k &\sim& \prod_{i=1}^{|{ k}|}  
	\frac{(\tau_i \sigma^2)^{\frac{r }{2}}}{\Gamma(\frac r 2)} 
	|\beta_{k_i}|^{-(r+1)} \exp\Big(-\frac{\tau_i\sigma^2}{\beta_{k_i}^2}\Big),\\ 
	\tau_i &\overset{i.i.d.}{\sim}&  \mbox{Inverse-Gamma }\Big(\frac{r + 1}2 , \psi\Big) , \quad \text{ for } i =1,\ldots, |k| ,
\end{eqnarray*}
where $\tilde{\tau}_k = (\tau_1,\tau_2,\ldots, \tau_{|k|})^\top$.
In particular, the authors put an independent piMOM prior on each linear regression coefficient (conditional on the hyperparameter $\tau_i$) and an Inverse-Gamma prior on $\tau_i$. 

There are some fundamental differences between \cite{BW:2017} and our work in terms of the models considered and corresponding analysis. 
Firstly, unlike the piMOM prior, the pMOM prior in our model does not in general correspond to assigning an independent prior to each entry of $ \beta_k$. 
In particular, pMOM distributions introduce correlations among the entries in $ \beta_k$ through $U_k$ and create more theoretical challenges. 
Furthermore, the pMOM prior imposes exact sparsity in $ \beta_k$, which is not the case for the piMOM prior in \cite{BW:2017}, thus they are structurally different.
Secondly, \cite{BW:2017} assumed the eigenvalues of the Hessian matrix to be bounded below and above by some constants, while we allow the upper bound to grow with $n$ (Condition \hyperref[cond_A2]{\rm (A2)}). 
In addition, to prove the model selection consistency, \cite{BW:2017} required the spectral norm of the difference between the Hessian matrices corresponding to any two models to be bounded above by a function of the $\ell_2$-norm difference between the respective regression coefficients, and they assumed that the product of the response variables and the entries of design matrix are bounded by a constant, while these constraints are not imposed in our study. 
See assumptions B1, B2 and C1 in \cite{BW:2017} for details. 
Thirdly, no simulation studies were conducted in \cite{BW:2017}, leaving the empirical validity of the proposed method in question, while we include the computational strategy in the following section and examine the practical utility of the hyper-pMOM prior in the context of gene expression analysis.

\section{Posterior Computation} \label{sec:computation}

In this section, we describe how to approximate the marginal density of data and conduct the model selection procedure. The integral formulation in \eqref{marginal} cannot be calculated in a closed form. 
Hence, we use Laplace approximation to compute $m_{ k}( y)$ and $\pi( k \mid  y)$. 
Similar approaches to compute posterior probabilities have been used in \cite{johnson2012bayesian}, \cite{Shi:nonlocal} and \cite{Shin2018Sinica}. 

For any model $ k$, when $ U_{ k} =  I_{ k}$, the normalization constant $d_{ k}$ in \eqref{model:pmom} is given by 
$d_{ k} = \big\{(2r-1)!!\big\}^{-| k|}$. Let 
%\begin{align*} 
%f( \beta_{ k}, \tau ) =& \log\Big(\exp \big\{L_n( {\beta_{k}}) \big\} \pi\left( {\beta_{k}} \mid  k, \tau\right)\pi(\tau)\Big) \\
%=& \sum_{i = 1}^n\Big\{y_i x_{i k}^\top {\beta_k} - \log\big(1 + \exp( x_{i k}^\top {\beta_k})\big)\Big\} - | k|\log\left((2r-1)!!\right) - \frac {| k|} 2\log(2\pi)  \\
%& - \Big(r| k| +\frac {| k|} 2\Big)\log\tau  -\frac{ \beta_{ k}^\top \beta_{ k}}{2\tau} + \sum_{i=1}^{| k|}2r\log\big(|\beta_{k_i}|\big)  - (\lambda_1 + 1)\log \tau - \frac{\lambda_2}\tau + \log\Big(\frac{\lambda_2^{\lambda_1}}{\Gamma(\lambda_1)}\Big).
%\end{align*}
\begin{align*} 
	f( \beta_{ k} ) =& \log\Big(\exp \big\{L_n( {\beta_{k}}) \big\} \pi( {\beta_{k}} \mid  k)\Big) \\
	=& \sum_{i = 1}^n\Big\{y_i x_{i k}^\top {\beta_k} - \log\big(1 + \exp( x_{i k}^\top {\beta_k})\big)\Big\} - | k|\log\left((2r-1)!!\right) - \frac {| k|} 2\log(2\pi)  + \log\Big(\frac{\lambda_2^{\lambda_1}}{\Gamma(\lambda_1)}\Big) \\
	& - \Big(r| k| +\frac {| k|} 2 + \frac{1}{2} \Big)\log \Big( \lambda_2 + \frac{1}{2} \|\beta_k\|_2^2 \Big)  
	+ 2r \sum_{i=1}^{| k|}\log\big(|\beta_{k_i}|\big)  + \log \Gamma \Big( r |k| + \frac{|k|}{2} + \lambda_1 \Big)    .
%	\\
%   \frac{\partial}{\partial \beta_k} f(\beta_k) =&  \sum_{i=1}^n x_{ik} \Big\{  y_i - \frac{ \exp( x_{i k}^\top {\beta_k})}{1+  \exp( x_{i k}^\top {\beta_k})}  \Big\} + \frac{2r}{\beta_k} 
%   - \Big( r|k| + \frac{|k|}{2} + \frac{1}{2} \Big) \frac{\beta_k}{\lambda_2 + \|\beta_k\|_2^2/2 }  \\
%  V  =& - \sum_{i = 1}^n\frac{ x_{i k} x_{i k}^\top\exp( x_{i k}^\top {\beta_k})}{\big\{1 + \exp( x_{i k}^\top {\beta_k})\big\}^2} - diag\bigg(\frac {2r} {\beta_{k_1}^2}, \ldots, \frac {2r} {\beta_{k_{| k|}}^2}\bigg) \\
%   & - \Big( r|k| + \frac{|k|}{2} + \frac{1}{2} \Big) \bigg\{ \frac{1}{\lambda_2 + \|\beta_k\|_2^2/2 } I_{|k|}  - \frac{1}{(\lambda_2 + \|\beta_k\|_2^2/2)^2 } \beta_k \beta_k^\top \bigg\}
\end{align*}
For any model $ k$, the Laplace approximation of $m_{ k}( y)$ is given by
\begin{align} \label{laplace}
	(2\pi)^{\frac {| k| } 2}\exp\big\{f(\what{ \beta}_{ k} )\big\}| V(\what{ \beta}_{ k} )|^{-\frac 1 2},
%(2\pi)^{\frac {| k| + 1} 2}\exp\big\{f(\what{ \beta}_{ k}, \what{\tau})\big\}| V(\what{ \beta}_{ k}, \what{
%\tau})|^{-\frac 1 2},
\end{align}
where $\what{ \beta}_{ k}  = \argmax_{ \beta_{ k}}f( \beta_{ k})$ is obtained via the optimization function \verb|optim| in \textsf{R} using a quasi-Newton method, and $ V( \beta_{ k} )$ is a $| k|\times | k| $ symmetric matrix defined as
\bea
V (\beta_k)  &=& - \sum_{i = 1}^n \frac{ x_{i k} x_{i k}^\top\exp( x_{i k}^\top {\beta_k})}{\big\{1 + \exp( x_{i k}^\top {\beta_k})\big\}^2} - {\rm diag}\bigg(\frac {2r} {\beta_{k_1}^2}, \ldots, \frac {2r} {\beta_{k_{| k|}}^2}\bigg) \\
&& - \Big( r|k| + \frac{|k|}{2} + \frac{1}{2} \Big) \bigg\{ \frac{1}{\lambda_2 + \|\beta_k\|_2^2/2 } I_{|k|}  - \frac{1}{(\lambda_2 + \|\beta_k\|_2^2/2)^2 } \beta_k \beta_k^\top \bigg\}   .
\eea
%\bea
%V_{11} &=& - \sum_{i = 1}^n\frac{ x_{i k} x_{i k}^\top\exp( x_{i k}^\top {\beta_k})}{\big\{1 + \exp( x_{i k}^\top {\beta_k})\big\}^2} - \frac 1 \tau  I_{ k} - diag\bigg(\frac {2r} {\beta_{k_1}^2}, \ldots, \frac {2r} {\beta_{k_{| k|}}^2}\bigg),\\
% V_{12} &=&  \frac{\beta_k}{\tau^2} ,  \quad V_{22} = -\frac{\beta_k^\top \beta_k + 2 \lambda_2}{\tau^3} + \frac{\lambda_1+1}{\tau^2}.  
%\eea
The above Laplace approximation can be used to compute the posterior probability ratio between two models.
%any given model $ k$ and the true model $ t$, and select the model $ k$ with the highest probability. 

%We then adopt 
The shotgun stochastic search (SSS) algorithm \citep{Hans2007,Shin2018Sinica} is inspired by MCMC but enables much more efficient identification of probable models by swiftly moving around in the model space as the dimension escalates. 
%Utilizing the Laplace approximations of the marginal probabilities in \eqref{laplace}, 
The SSS algorithm explores high-dimensional model spaces and quickly identifies ``interesting'' regions of high posterior probability over models. 
The SSS evaluates numerous models guided by the unnormalized posterior probabilities that can be approximated using the Laplace approximations of the marginal probabilities in (\ref{laplace}). 
Let $\mbox{nbd}( k) = \{\Gamma_k^+, \Gamma_k^-, \Gamma_k^0\}$ containing all the neighbors of model $ k$, in which $\Gamma_k^+ = \big\{ k \,\cup \{j\}: j \notin  k\big\}$, $\Gamma_k^- = \big\{ k\setminus\{j\}: j \in  k\big\}$ and $\Gamma_k^0 = \big\{ k\setminus\{j\} \cup \{l\}: j \in  k, l \notin  k\big\}$. 
Algorithm \ref{CHalgorithm} describes the SSS procedure.
\begin{algorithm}[tb]
	\caption{Shotgun Stochastic Search (SSS)}
	\label{CHalgorithm}
	\begin{algorithmic}
		\State Set an initial model $ k^{(1)}$
		\For{$i = 1$ to $i = N - 1$}
		\State (a) Compute $\pi( k\mid y)$ using \eqref{laplace} for all $ k \in \mbox{nbd}\big( k^{(i)}\big)$
		\State (b) Sample $ k^+$, $ k^-$ and $ k^0$ from $\Gamma_k^+$, $\Gamma_k^-$ and $\Gamma_k^0$ with probabilities proportional to $\pi( k \mid  y)$
		\State (c) Sample the next model $ k^{(i+1)}$ from $\{ k^{+},  k^{-},  k^{0}\}$ with probability proportional to 
		\State \quad\,  $\big\{\pi( k^+ \mid  y), \pi( k^- \mid  y), \pi( k^0 \mid  y)\big\}$
		\EndFor
	\end{algorithmic}
\end{algorithm}

However, as pointed out by \cite{Shin2018Sinica}, the SSS algorithm can be computationally expensive in high-dimensional settings.
The computational bottleneck is calculating the Laplace approximations of the marginal probabilities for the models in $\Gamma_k^+$ and $\Gamma_k^0$, whose cardinalities are $p-|k|$ and $(p-|k|)|k|$, respectively.
To alleviate computational burden, we slightly modify the SSS algorithm by reducing the number of entries in $\Gamma_k^+$ and $\Gamma_k^0$.
Specifically, we reduce the number of models in $\Gamma_k^+$ by selecting only (1) the top $K_1$ variables having large absolute sample correlation with $y$ and (2) $K_2$ randomly selected variables, and we define the resulting set as $\Gamma^+_{R,k}$.
Similarly, we define a reduced set $\Gamma^0_{R,k} = \{ k \setminus \{j\} \cup \{l\} : j \in k , l \in \Gamma^+_{R,k}  \}$ and replace $\Gamma^0_k$ with $\Gamma^0_{R,k}$ in the algorithm. 
By doing so, we can efficiently reduce the computational complexity of the algorithm.
Note that the cardinalities of $\Gamma^+_{R,k}$ and $\Gamma^0_{R,k}$ are $K_1+K_2$ and $(K_1+K_2) |k|$, respectively.
We call this modified algorithm the reduced SSS (RSSS) algorithm and describe it in Algorithm \ref{RCHalgorithm}.
Note that the RSSS algorithm is different from the simplified shotgun stochastic search with screening (S5) algorithm \citep{Shin2018Sinica}.
The two main differences are that the RSSS algorithm does not completely ignore the set $\Gamma^0_k$ and does not introduce temperature parameters. 
In the subsequent simulation study and real data analysis, the RSSS algorithm with $K_1=K_2=10$ is adopted for the posterior inference of the hyper-pMOM prior.
\begin{algorithm}[tb]
	\caption{Reduced Shotgun Stochastic Search (RSSS)}
	\label{RCHalgorithm}
	\begin{algorithmic}
		\State Set an initial model $ k^{(1)}$
		\For{$i = 1$ to $i = N - 1$}
		\State (a) Compute $\pi( k\mid y)$ using \eqref{laplace} for all $ k \in \mbox{nbd}_R\big( k^{(i)}\big)= \{\Gamma_{R,k}^+, \Gamma_k^-, \Gamma_{R,k}^0\}$
		\State (b) Sample $ k^+$, $ k^-$ and $ k^0$ from $\Gamma^+_{R,k}$, $\Gamma_k^-$ and $\Gamma^0_{R,k}$ with probabilities proportional to $\pi( k \mid  y)$
		\State (c) Sample the next model $ k^{(i+1)}$ from $\{ k^{+},  k^{-},  k^{0}\}$ with probability proportional to 
		\State \quad\,  $\big\{\pi( k^+ \mid  y), \pi( k^- \mid  y), \pi( k^0 \mid  y)\big\}$
		\EndFor
	\end{algorithmic}
\end{algorithm}

\begin{figure}
	\centering
	\includegraphics[width=0.53\linewidth]{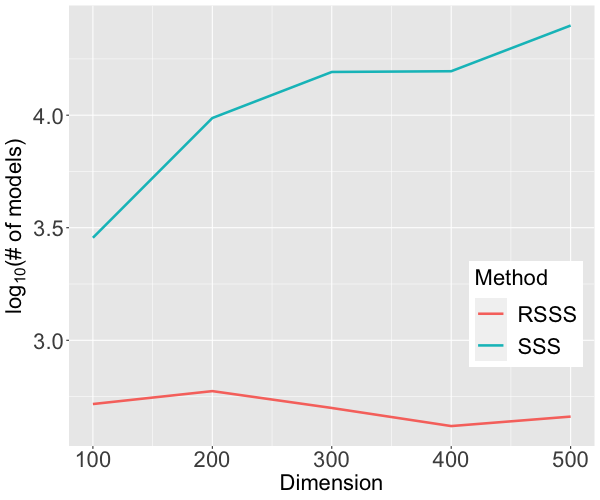}
	\caption{The average number of models searched before visiting the posterior mode.}
	\label{fig:RSSS_SSS}
\end{figure}

To demonstrate the computational efficiency of the RSSS algorithm and compare it with the SSS algorithm, we conduct a simulation study.
We generate the data from the model \eqref{logistic_model} with the true coefficient $\beta_0 = (1,1,1,0,\ldots, 0)^\top \in \bbR^p$ and design matrix $X= (x_1,\ldots, x_n)^\top \in \bbR^{n\times p}$, where $x_i \overset{i.i.d.}{\sim} N_p(0, I_p)$ for $i=1,\ldots, n$.
The number of samples is fixed at $n=100$, while the number of variables varies over $ p \in \{100,\ldots, 500\}$.
Figure \ref{fig:RSSS_SSS} shows the average number of models searched before visiting the posterior mode for each $p$, where the averages are calculated based on 10 repetitions.
When compared with the SSS algorithm, the RSSS algorithm investigates a far less number of models before hitting the posterior mode, while both algorithms found the same posterior mode for all data sets in our simulations.
%We also found that the maximum values of the posterior score \eqref{laplace} obtained by the two methods are quite similar: the maximum difference between the two is less than 0.038.
Therefore, the RSSS algorithm can achieve nearly identical performance to the SSS algorithm while boosting computing efficiency.

\section{Simulation Studies}\label{sec:sim}

In this section, we investigate the performance of the hyper-pMOM prior for logistic regression models.
For given $n=100$ and $p \in \{100, 300\}$, 
simulated data sets are generated from \eqref{logistic_model} with the true coefficient vector $\beta_0$ and design matrix $X= (x_1,\ldots, x_n)^\top \in \bbR^{n\times p}$. 
%Let $t$ be the index set for the nonzero values in $\beta_0$.
We set the index set for nonzero values in $\beta_0$ at $t = \{1,2,3\}$, where nonzero coefficients $\beta_{0, t}$ are generated under the following two different settings:
\begin{itemize}
	\item Setting 1 (Weak signals): All the entries of $\beta_{0,t}$ are set to $1$.
	\vspace{-.5cm}
	
	\item Setting 2 (Moderate signals): All the entries of $\beta_{0,t}$ are set to $2$.
\end{itemize}
We generate covariate vectors as $x_i \overset{i.i.d.}{\sim} N_p(0, \sg)$ for $i=1,\ldots, n$, under the following cases of $\sg$:
\begin{itemize}
	\item Case 1 (Isotropic design): $\sg = I_p$ \vspace{-.5cm}
	
	\item Case 2 (Correlated design):  $\sg = (\sg_{ij})$, where $\sg_{ij} = 0.3^{|i-j|}$ for any $1\le i \le j \le p$.
\end{itemize}
We also generate test samples $\{(y_{{\rm test}, 1}, x_{{\rm test}, 1}),\ldots, (y_{{\rm test}, n_{\rm test}}, x_{{\rm test}, n_{\rm test}}) \}$ with $n_{\rm test}=50$ to evaluate the prediction performance.

In the various scenarios mentioned above, we compare the performance of our method with existing variable selection methods.
As Bayesian contenders, we consider the nonlocal pMOM prior \citep{caoEntropy2020}, 
spike and slab prior \citep{SpikeSlab} and
empirical Bayesian Lasso (EBLasso) \citep{EBLasso}, 
while we consider Lasso \citep{Lasso} and and SCAD \citep{SCAD} as frequentist competitors.

The \textsf{R} codes for implementing the hyper-pMOM prior are publicly available at \url{https://github.com/leekjstat/Hierarchical-nonlocal}. 
The hyperparameters in \eqref{model:pmom} and \eqref{model:tau} are set at $U = I_p$, $r=1$, $\lambda_1=1$ and $\lambda_2 = 10^2 n^{-1/3}p^{(2+ 0.001)2/3}$, which satisfy Condition \hyperref[cond_A4]{\rm (A4)}.
For the implementation of the pMOM prior, the \textsf{R} codes available at \url{https://github.com/xuan-cao/Nonlocal-Logistic-Selection} are used, where the hyperparameters are set at $U = I_p$, $r=1$ and $\tau = n^{-1/2} p^{2+0.05}$.
For both RSSS and SSS procedures, initial models are set by randomly taking three nonzero entries.
For the regularization approaches, the tuning parameters are chosen by 5-fold cross-validation.

To examine the performance of each method, the values of the precision, sensitivity, specificity, Matthews correlation coefficient (MCC) \citep{matthews1975comparison} and mean squared prediction error (MSPE) are used.
These criteria are defined as
\bea
&&\text{Precision}  =    \frac{TP}{TP+FP} , \quad
\text{Sensitivitiy}  =     \frac{TP}{TP+FN} ,  \quad
\text{Specificity}  =   \frac{TN}{TN+FP}   ,  \\
&&\text{MCC} =   \frac{TP \times TN - FP\times FN}{\sqrt{(TP+FP)(TP+FN)(TN+FP)(TN+FN)}}, \quad	  
\text{MSPE}  =  \frac{1}{n_{\rm test}} \sum_{i=1}^{n_{\rm test}}     \big( \hat{y}_i - y_{{\rm test},i}  \big)^2,
\eea
where \emph{TP}, \emph{TN}, \emph{FP} and \emph{FN} are true positive, true negative, false positive and false negative, respectively.
Here, $\hat{y}_{i} = \exp  (x_{{\rm test}, i}^\top \hat{\beta})/ \{ 1+ \exp  (x_{{\rm test}, i}^\top \hat{\beta}) \}$, where $\hat{\beta}$ is the estimated coefficient vector.
For the hyper-pMOM  and pMOM priors, the nonzero part in $\hat{\beta}$ is chosen as the posterior mode with the estimated model $\hat{k}$, i.e., $\hat{\beta}_{\hat{k}} = \argmax_{\hat{k}} f( \beta_{\hat{k}} )$.
For the spike and slab prior, the posterior mean based on 2,000 posterior samples is used as $\hat{\beta}$.
The averages of each criterion based on 10 repetitions are summarized in Tables \ref{table:comp1}--\ref{table:comp4}.

\begin{table}[!tb]
	\centering
	\caption{
		The summary statistics for Case 1 (isotropic design) when $p=100$.
	}
	\begin{tabular}{c c c c c c c}
		\toprule
%		& \multicolumn{5}{c}{ \textbf{Setting 1} } \\  \midrule
		& &\textbf{Precision} & \textbf{Sensitivity} & \textbf{Specificity} & \textbf{MCC} & \textbf{MSPE}  \\ \midrule
		\multirow{6}{*}{Setting 1} &  Hyper-pMOM & \bf 1.000  &0.667 & \bf 1.000 & \bf 0.812 &0.210  \\ 
		& pMOM & \bf 1.000 &0.667 & \bf 1.000 & \bf 0.812& \bf 0.206  \\ 
		& Spike and slab  & 0.917 &0.700 &0.997 &0.744 &0.227  \\ 
		& EBLasso &0.500 &0.733 & 0.979&0.618 &0.256  \\ 
		& Lasso  & 0.364 &0.900 &0.889 &0.438 &0.210  \\ 
		& SCAD & 0.389 & \bf 0.900 &0.915 &0.470 &0.208  \\ \midrule
%		& \multicolumn{5}{c}{ \textbf{Setting 2} }  \\ \midrule%unif(1,5, 3)
%		&\textbf{Precision} & \textbf{Sensitivity} & \textbf{Specificity} & \textbf{MCC} & \textbf{MSPE}  \\ \midrule
		\multirow{6}{*}{Setting 2} & Hyper-pMOM  & \bf 1.000 &0.967 &\bf 1.000 &\bf 0.981 &0.121  \\ 
		& pMOM  &0.933 &\bf 1.000 &0.994 &0.956 &0.132  \\ 
		& Spike and slab  &0.832 &\bf 1.000 &0.990 &0.899 &0.124  \\ 
		& EBLasso &0.920  &\bf 1.000 &0.996 &0.953 &0.227  \\ 
		& Lasso  &0.241 &\bf 1.000 &0.877 &0.453 &\bf 0.147  \\ 
		& SCAD &0.337 &\bf 1.000 &0.929 &0.533 &0.125  \\  \bottomrule
	\end{tabular}\label{table:comp1}
\end{table}
\begin{table}[!tb]
	\centering
	\caption{
		The summary statistics for Case 2 (correlated design) when $p=100$.
	}
	\begin{tabular}{c c c c c c c}
		\toprule
%		& \multicolumn{5}{c}{ \textbf{Setting 1} } \\  \midrule
		& &\textbf{Precision} & \textbf{Sensitivity} & \textbf{Specificity} & \textbf{MCC} & \textbf{MSPE}  \\ \midrule
		\multirow{6}{*}{Setting 1} & Hyper-pMOM  &\bf 1.000  &0.667 &\bf 1.000 & 0.812 &0.177  \\ 
		&pMOM &\bf 1.000 &0.667 &\bf 1.000 & 0.812&0.177  \\ 
		&Spike and slab  & 0.927 &0.867 &0.997 &\bf 0.881 &0.182  \\ 
		&EBLasso &0.753 &0.800 & 0.986&0.749 &0.235  \\ 
		&Lasso  & 0.392 &\bf 0.967 &0.920 &0.566 &0.172  \\ 
		&SCAD & 0.452 & \bf 0.967 &0.938 &0.619 &\bf 0.167  \\ \midrule
%		& \multicolumn{5}{c}{ \textbf{Setting 2} }  \\ \midrule%unif(1,5, 3)
%		&\textbf{Precision} & \textbf{Sensitivity} & \textbf{Specificity} & \textbf{MCC} & \textbf{MSPE}  \\ \midrule
		\multirow{6}{*}{Setting 2} &Hyper-pMOM &\bf 0.975 &0.800 &\bf 0.999 &0.874 &0.134  \\ 
		&pMOM &0.950 &0.833 &0.998 &0.878 &0.135  \\ 
		&Spike and slab  &0.835 & 0.967 &0.992 &0.886 &0.116  \\ 
		&EBLasso &0.900  & 0.967 &0.996 &\bf 0.926 &0.208  \\ 
		&Lasso  &0.220 &\bf 1.000 &0.870 &0.433 &0.125  \\ 
		&SCAD &0.307 &\bf 1.000 &0.921 &0.528 &\bf 0.113  \\  \bottomrule
	\end{tabular}\label{table:comp2}
\end{table}
\begin{table}[tb]
	\centering
	\caption{
		The summary statistics for Case 1 (isotropic design) when $p=300$.
	}
	\begin{tabular}{c c c c c c c}
		\toprule
%		& \multicolumn{5}{c}{ \textbf{Setting 1} } \\  \midrule
		& &\textbf{Precision} & \textbf{Sensitivity} & \textbf{Specificity} & \textbf{MCC} & \textbf{MSPE}  \\ \midrule
		\multirow{6}{*}{Setting 1} &Hyper-pMOM & 0.900   &    0.600  &     \bf 0.999 & \bf 0.733 &  0.227  \\ 
		&pMOM & 0.850    &   0.567    &   \bf 0.999 &0.692& 0.232  \\ 
		&Spike and slab  & \bf 0.925   &    0.567  &    \bf  0.999& 0.661 &0.238 \\ 
		&EBLasso & 0.660&       0.533 &      0.996& 0.563& 0.252  \\ 
		&Lasso  & 0.191  &     \bf 0.967    &   0.949 &0.410& \bf 0.215  \\ 
		&SCAD & 0.199      & \bf 0.967     &  0.953 &0.420& 0.219  \\ \midrule
%		& \multicolumn{5}{c}{ \textbf{Setting 2} }  \\ \midrule%unif(1,5, 3)
%		&\textbf{Precision} & \textbf{Sensitivity} & \textbf{Specificity} & \textbf{MCC} & \textbf{MSPE}  \\ \midrule
		\multirow{6}{*}{Setting 2} &Hyper-pMOM& \bf 1.000    &   0.800  &   \bf  1.000 &0.889 &0.164  \\ 
		&pMOM & 0.943      & 0.800   &    0.999 &0.854 &0.166  \\ 
		&Spike and slab  & 0.875 &      0.967      & 0.998& \bf 0.917& 0.150  \\ 
		&EBLasso & 0.925 &      0.900 &      0.999& 0.901& 0.224  \\ 
		&Lasso  & 0.184   &  \bf  1.000  &     0.941 &0.406 &0.153  \\ 
		&SCAD & 0.198     &\bf  1.000    &   0.956 &0.433& \bf 0.139  \\  \bottomrule
	\end{tabular}\label{table:comp3}
\end{table}
\begin{table}[tb]
	\centering
	\caption{
		The summary statistics for Case 2 (correlated design) when $p=300$.
	}
	\begin{tabular}{c c c c c c c}
		\toprule
%		& \multicolumn{5}{c}{ \textbf{Setting 1} } \\  \midrule
		& &\textbf{Precision} & \textbf{Sensitivity} & \textbf{Specificity} & \textbf{MCC} & \textbf{MSPE}  \\ \midrule
		\multirow{6}{*}{Setting 1} &Hyper-pMOM &\bf 1.000  &0.667 &\bf 1.000 & 0.812 &0.177  \\ 
		&pMOM &\bf 1.000 &0.667 &\bf 1.000 & 0.812&0.177  \\ 
		&Spike and slab  & 0.927 &0.867 &0.997 &\bf 0.881 &0.182  \\ 
		&EBLasso &0.753 &0.800 & 0.986&0.749 &0.235  \\ 
		&Lasso  & 0.392 &\bf 0.967 &0.920 &0.566 &0.172  \\ 
		&SCAD & 0.452 & \bf 0.967 &0.938 &0.619 &\bf 0.167  \\ \midrule
%		& \multicolumn{5}{c}{ \textbf{Setting 2} }  \\ \midrule%unif(1,5, 3)
%		&\textbf{Precision} & \textbf{Sensitivity} & \textbf{Specificity} & \textbf{MCC} & \textbf{MSPE}  \\ \midrule
		\multirow{6}{*}{Setting 2} &Hyper-pMOM &\bf 1.000  &     0.667 &   \bf   1.000& \bf 0.815 &0.191 \\ 
		&pMOM & 0.950     &  0.633  &    \bf  1.000 &0.774 & \bf 0.184  \\ 
		&Spike and slab  & 0.908&       0.700     &  0.999& 0.779 &0.203  \\ 
		&EBLasso & 0.708   &    0.767 &      0.995 &0.703& 0.240  \\ 
		&Lasso  & 0.366     &  \bf 0.967   &    0.969& 0.555& 0.187  \\ 
		&SCAD & 0.331      & \bf 0.967  &     0.973 &0.539 &0.187  \\  \bottomrule
	\end{tabular}\label{table:comp4}
\end{table}

Based on the simulation results, the proposed hyper-pMOM prior tends to achieve high precision and specificity in most settings, which means that the hyper-pMOM prior produces low false positives.
The hyper-pMOM prior also tends to have high MCC especially under the isotropic design (Case 1), and even under the correlated design (Case 2), overall higher MCC than the EBLasso and frequentist methods.
Furthermore, the hyper-pMOM prior outperforms the pMOM prior in the majority of situations, demonstrating its relative superiority to the pMOM prior.
%Furthermore, interestingly, the hierarchical nonlocal prior outperforms the nonlocal prior in most cases, which supports its relative advantage over the nonlocal prior.
Overall, the Bayesian methods achieve high precision, specificity and MCC, while the frequentist methods have high sensitivity and MSPE.
Similar observations have also been discussed in the literature \citep{meinshausen2006,caoEntropy2020}.

\section{Application to the Analysis of Differentially Expressed Genes}\label{sec:real}

Asthma has been recognized as a systemic disease consisting of networks of genes showing inflammatory changes involving a broad spectrum of adaptive and innate immune systems. Utilizing measurable characteristics of asthmatic patients, including biologic gene expression markers, can help to identify phenotypic categories in asthma. Identification of these phenotypes may help develop strategies for preventing progression of disease severity \citep{Carr2016}. 
We aim to apply the proposed variable selection method to develop an RNA-seq-based risk score for asthma stratification.

To construct the risk score, gene expression analysis is performed using an asthma RNA-seq dataset GSE146046 in the Gene Expression Omnibus (GEO) database \citep{GSE}. There are 95 individuals in the GSE146046 dataset including 51 asthmatic subjects and 44 non-asthmatic subjects. The gene expression levels of all the 95 individuals are first randomly split into 2/3 as training and 1/3 as test data while maintaining the same ratio between asthma and control groups.  %asthma:control ratio. 
Next we conduct the analysis of differentially expressed genes (DEG) based on the training set and construct data tables containing raw count values for approximately 20,000 unique genes, with genes in rows and sample GEO accession numbers in columns. 
\verb|DESeq2| \textsf{R} package is used to store the read counts and the intermediate estimated quantities during statistical analysis \citep{Love2014}. We extract summary statistics including \textit p-values for all genes and retain a total of 180 DEGs with \textit p-values less than $10^{-4}$ visualized in a Manhattan plot (Figure \ref{fig1}). 
The proposed method and other contenders are applied to the resulting dataset with $p = 180$. The hyperparameters for all the methods are set as in the simulation studies.

\begin{figure}[tb]
	\centering
	\includegraphics[width=0.75\linewidth]{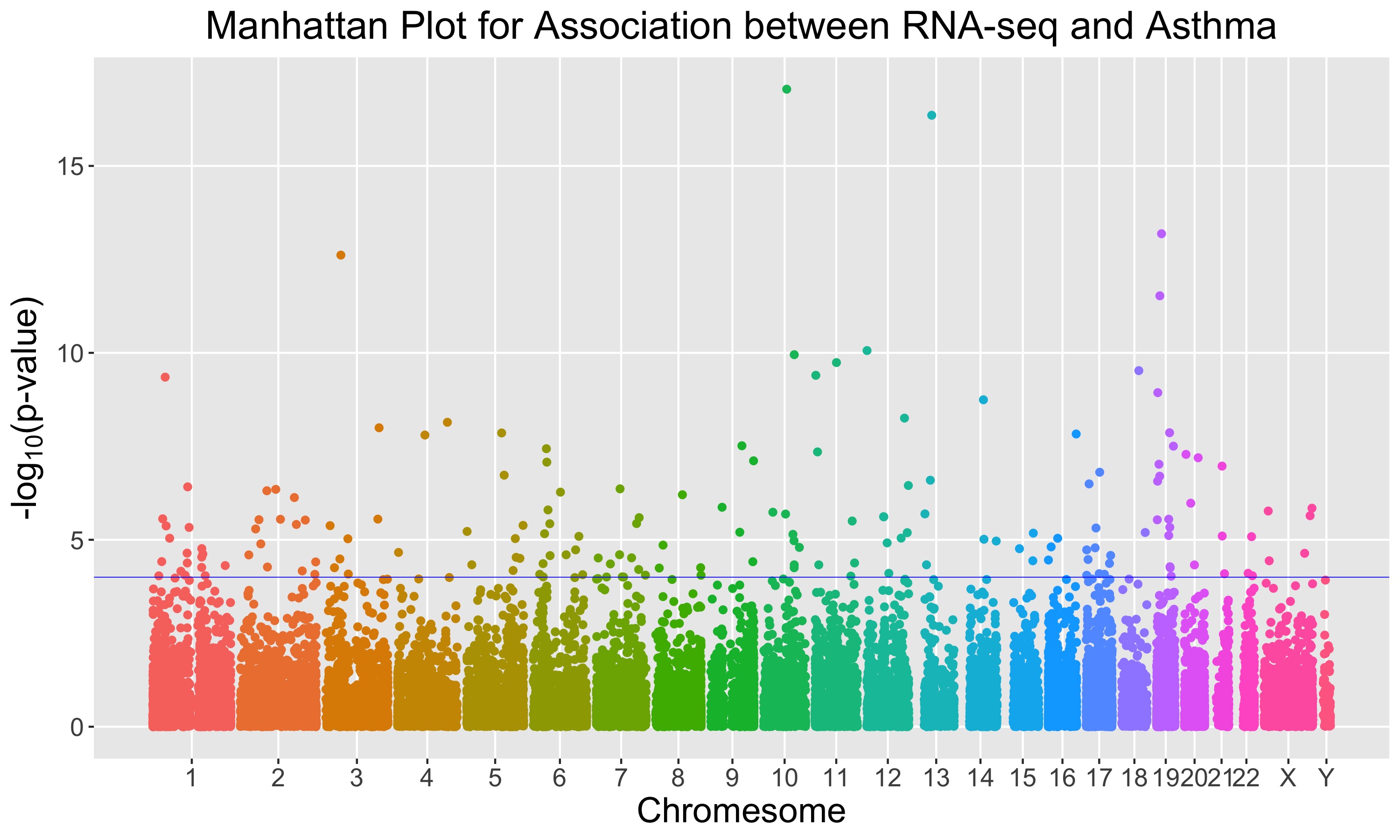}
	\caption{Manhattan plot for association between RNA-seq and asthma. The points beyond the blue line represents the 180 selected genes with $p < 10^{-4}$ based on the DEG analysis.}
	\label{fig1}
\end{figure}

In Figure \ref{fig2}, we draw the receiver operating characteristic (ROC) curves for all the methods. The results are further summarized in Table \ref{table:comp5} where a common cutoff value 0.5 is adopted for thresholding prediction. From Table \ref{table:comp5} and Figure \ref{fig2}, we can tell that the hyper-pMOM prior has overall better prediction performance compared with other methods. Of the 180 genes, four genes, namely, TRIM26, CAPSL, FOXA3 and PYY, are selected by the proposed method. These identified genes seem plausible and have been established in the asthma GWAS catalog \citep{SCHOETTLER20191495}, which may help better understand the omics architecture that drives complex diseases.

\begin{figure}[tb]
	\centering
	\includegraphics[width=0.53\linewidth]{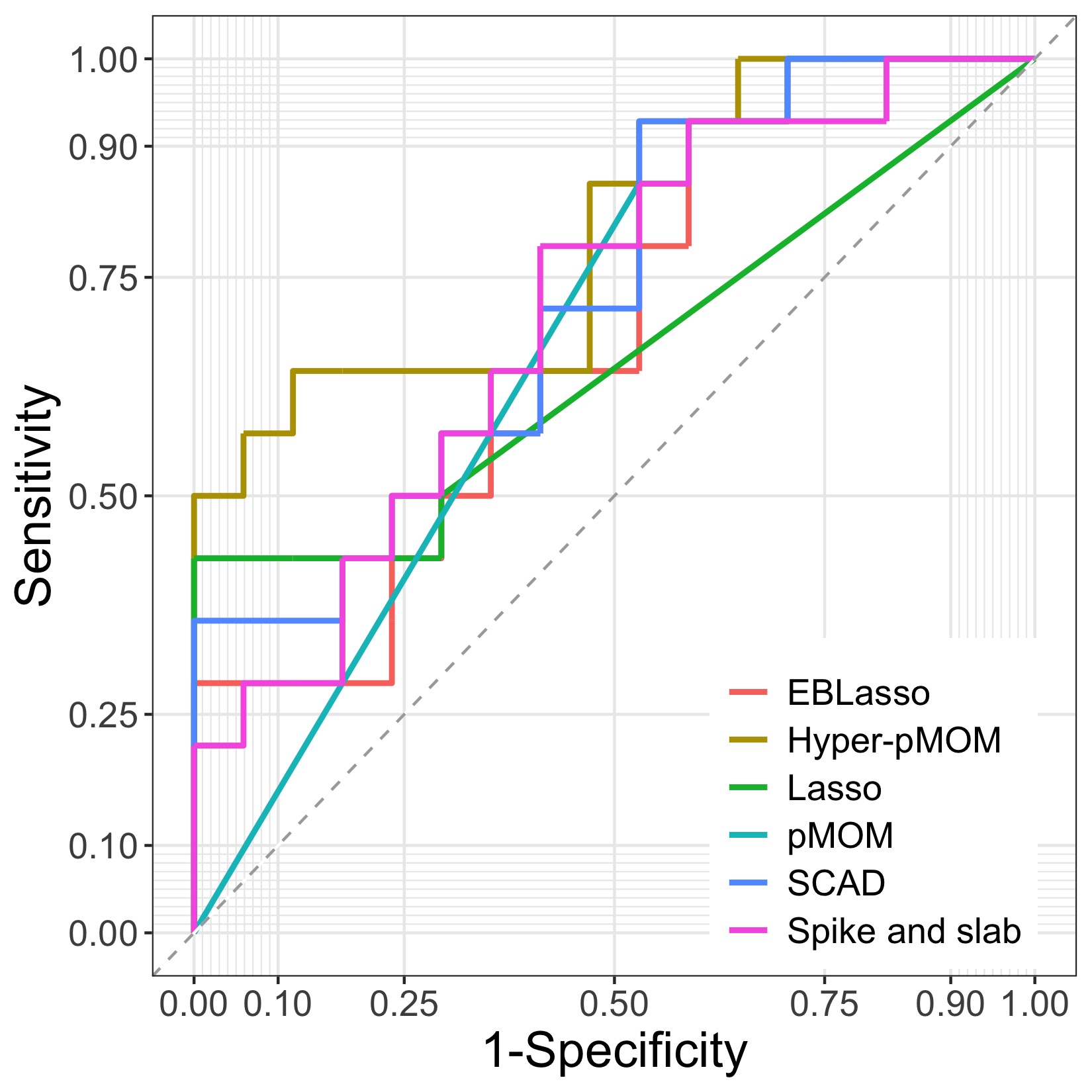}
	\caption{ROC curves comparison between different methods. X-axis: false positive rate (1 - specificity). Y-axis: true positive rate (sensitivity).}
	\label{fig2}
\end{figure}
\begin{table}[tb]
	\centering
	\caption{
		The summary statistics for prediction performance in the testing set.
	}
	\begin{tabular}{c c c c c c c}
		\toprule
		%		& \multicolumn{5}{c}{ \textbf{Setting 1} } \\  \midrule
		&	&\textbf{Precision} & \textbf{Sensitivity} & \textbf{Specificity} & \textbf{MCC} & \textbf{MSPE}  \\ \midrule
		&Hyper-pMOM &\bf 0.750 &0.643 &\bf 0.824 &\bf 0.477&0.225  \\ 
		&pMOM &0.565 &\bf 0.929 &0.412 &0.387&0.355  \\ 
		&Spike and slab  &0.571 &0.857&0.471 &0.349&0.233 \\ 
		&EBLasso &0.583 &0.500 &0.706 &0.210&0.238  \\ 
		&Lasso  &0.545 &0.857 &0.412&0.295 &0.233 \\ 
		&SCAD &0.565  &0.929&0.412&0.387 &\bf 0.224  \\  \bottomrule
	\end{tabular}\label{table:comp5}
\end{table}

\section{Discussion}\label{sec:disc}
In this paper, we consider the hyper-pMOM prior and investigate asymptotic properties of the resulting posterior distribution.
Although the hyper-pMOM prior possesses thicker tails than the pMOM prior by adopting the hyperprior \eqref{model:tau} rather than a fixed $\tau$, it still has the hyperparameters, $\lambda_1$ and $\lambda_2$.
Because the choice of hyperparameters could affect variable selection performance, a cross-validation-based selection approach for $\lambda_1$ and $\lambda_2$ will be worth exploring.
In this case, studying the theoretical properties of the posterior based on the hyperparameters chosen by the cross-validation procedure will be a challenging but important task.

Furthermore, as mentioned in Section \ref{sec:selection cons}, deriving strong model selection consistency in a broader class of GLMs is an interesting future research direction.
Note that, in this work, we focus on logistic regression models when proving strong model selection consistency of the posterior.
An extension to general GLMs might require more conditions on the design matrix based on the current techniques used in the proof, due to more complicated structure of the Hessian matrix for other GLMs compared with that for the logistic regression model.
%deriving strong model selection consistency in general GLMs can serve as one of the interesting future research directions.

\newpage

\section*{Supplementary Material}

Throughout the Supplementary Material, we assume that for any  
$$ u \in \{ u \in \bbR^n :  u \text{ is in the space spanned by the columns of }  \sg^{1/2}  X_{ k} \}$$ 
and any model $ k \in \{ k \subseteq [r] : | k| \le m_n + | t| \}$, there exists $\delta^*>0$ such that
\bean\label{delta_star}
\bbE \Big[ \exp \big\{ u^\top  \sg^{-1/2}( y -  \mu) \big\} \Big] &\le& \exp\Big\{\frac{(1+\delta^*) u^\top  u}{2} \Big\} ,
\eean
for any $n \ge N(\delta^*)$.
However, as stated in \cite{Naveen:2018}, there always exists $\delta^*>0$ satisfying inequality \eqref{delta_star}, so it is not really a restriction.
Since we will focus on sufficiently large $n$, $\delta^*$ can be considered an arbitrarily small constant, so  we can always assume  that $\delta > \delta^*$.

\begin{proof}[Proof of Theorem \ref{thm:nosuper}]
	Let $M_1 = \{ k:  k \supsetneq  t, | k| \le m_n \}$ and
	\bea
	PR( k,  t) &=& \frac{\pi( k \mid  y)}{\pi( t \mid  y)},
	\eea
	where $t \subseteq [r]$ is the true model.
	We will show that 
	\bean\label{sel_goal1}
	\sum_{ k:  k \in M_1} PR( k,  t) &\overset{P}{\lra}& 0 \quad \text{ as } n\to\infty.
	\eean
	
%	\noindent{\bf Super sets.}
	By Taylor's expansion of $L_ n( {\beta_{k}})$ around $\what{ \beta}_{ k}$, which is the MLE of $ {\beta_{k}}$ under the model $ k$, we have
	\bea
	L_n( {\beta_{k}}) &=& L_n (\what{ \beta}_{ k}) - \frac{1}{2}( {\beta_{k}} - \what{ \beta}_{ k} )^\top  H_n(\tilde{ \beta}_{ k}) ( {\beta_{k}} - \what{ \beta}_{ k} ) 
	\eea
	for some $\tilde{ \beta}_{ k}$ such that $\|\tilde{ \beta}_{ k} - \what{ \beta}_{ k}\|_2 \le \| {\beta_{k}} - \what{ \beta}_{ k}\|_2$.
	Furthermore, by Lemmas A.1 and A.3 in \cite{Naveen:2018} and Condition \hyperref[cond_A2]{\rm (A2)}, with probability tending to 1,
	\bea
	- \frac{1+\epsilon}{2} ( {\beta_{k}} - \what{ \beta}_{ k} )^\top  H_n( \beta_{0, k}) ( {\beta_{k}} - \what{ \beta}_{ k} ) \le	L_n( {\beta_{k}}) - L_n (\what{ \beta}_{ k}) 
	\le - \frac{1-\epsilon}{2} ( {\beta_{k}} - \what{ \beta}_{ k} )^\top  H_n( \beta_{0, k}) ( {\beta_{k}} - \what{ \beta}_{ k} ) 
	\eea
	for any $ k \in M_1$ and $ {\beta_{k}}$ such that $\| {\beta_{k}} -  \beta_{0,  k}\|_2 < c \sqrt{ | k| \Lambda_{| k|} \log p /n } = : c w_n$, where $\epsilon = \epsilon_n := c'\sqrt{m_n^2 \Lambda_{m_n} \log p /n } = o(1)$, for some constants $c,c'>0$.
	Note that for $ {\beta_{k}}$ such that $\| {\beta_{k}} - \what{ \beta}_{ k}\|_2 = cw_n /2$, 
	\bea
	L_n( {\beta_{k}}) - L_n (\what{ \beta}_{ k})  
	&\le& - \frac{1-\epsilon}{2} \, \| {\beta_{k}} - \what{ \beta}_{ k}\|_2^2 \, \lambda_{\min} \big\{  H_n( \beta_{0, k})  \big\} \\
	&\le& - \frac{1-\epsilon}{2} \, \frac{c^2 w_n^2}{4} n \lambda 
	= - \frac{1-\epsilon}{8} c^2 \lambda | k| \Lambda_{| k|} \log p \,\, \lra \,\, - \infty \quad \text{ as }n\to\infty,
	\eea
	where the second inequality holds due to Condition \hyperref[cond_A2]{\rm (A2)}.
	It also holds for any $ {\beta_{k}}$ such that $\| {\beta_{k}} - \what{ \beta}_{ k}\|_2 > cw_n /2$ by concavity of $L_n(\cdot)$ and the fact that $\what{ \beta}_{ k}$ maximizes $L_n( {\beta_{k}})$.
	
	Define the set $B := \big\{  {\beta_{k}}:  \| {\beta_{k}} - \what{ \beta}_{ k}\|_2 \le cw_n /2  \big\},$
	then we have $B \subset \{ {\beta_{k}}: \| {\beta_{k}}-  \beta_{0, k}\|_2 \le cw_n  \}$ for some large $c>0$ and any $ k \in M_1$, with probability tending to 1.	
	\bean \label{marginal density}
		m_{ k}( y) &=&\int \int \exp \big\{L_n( {\beta_{k}}) \big\} \pi\left( {\beta_{k}} \mid  k\right)\,d {\beta_{k}} d\tau \nonumber\\
		&=& \int \int \exp \big\{L_n( {\beta_{k}}) \big\}  d_{ k}(2\pi)^{-| k|/2}\tau^{-r| k| - | k|/2} | U_{ k}| ^{\frac 1 2} \exp \Big(- \frac{ {\beta_{k}} ^\top  U_{ k}  {\beta_{k}}}{2 \tau}\Big)\prod_{i =1}^{| k|} \beta_{k_i}^{2r} \nonumber\\
		&&\qquad \times \,\,\frac{\lambda_2^{\lambda_1}}{\Gamma(\lambda_1)}\tau^{-\lambda_1 - 1} \exp\Big(-\frac{\lambda_2}\tau\Big)\, d {\beta_{k}} d\tau  \nonumber\\
		&\le&  \frac{\lambda_2^{\lambda_1}}{\Gamma(\lambda_1)}d_{ k}(2\pi)^{-| k|/2} | U_{ k}| ^{\frac 1 2} \exp \big\{ L_n(\what{ \beta}_{ k}) \big\}\nonumber\\
		&&\times\,\, \Bigg\{ \int \int_B M_1 \,d {\beta_{k}} d\tau  +\exp \big( -\frac{1-\epsilon}{8} c^2 \lambda | k| \Lambda_{| k|} \log p  \big) \int \int_{B^c} M_2\, d {\beta_{k}} d\tau \Bigg\},
	\eean 
	where 
	\bea
	M_1 = \exp\Big(-\frac{\lambda_2}\tau\Big) \exp \bigg\{ -\frac{1-\epsilon}{2} ( {\beta_{k}} - \what{ \beta}_{ k} )^\top  H_n( \beta_{0, k}) ( {\beta_{k}} - \what{ \beta}_{ k} ) - \frac{ {\beta_{k}} ^\top  U_{ k}  {\beta_{k}}}{2 \tau}  \bigg\} \prod_{i =1}^{| k|} \beta_{k_i}^{2r} 
	\eea
	and 
	\bea
	M_2 = \tau^{-r| k| - | k|/2 - \lambda_1 - 1} \exp\Big(-\frac{\lambda_2}\tau\Big) \exp \bigg(- \frac{ {\beta_{k}} ^\top  U_{ k}  {\beta_{k}}}{2 \tau} \bigg) \prod_{i =1}^{| k|} \beta_{k_i}^{2r}.
	\eea
	Note that for $ A_{ k} = (1-\epsilon) H_n( \beta_{0, k})$ and $ {\beta_{k}}^* = ( A_{ k} +  U_{ k}/\tau)^{-1}  A_{ k} \what{ \beta}_{ k}$, we have
	\bea
	&& \int \int_B M_1\, d {\beta_{k}} d\tau \\
	&\le& \int \int \tau^{-r| k| - | k|/2 - \lambda_1 - 1}  \exp\Big(-\frac{\lambda_2}\tau\Big)\exp \bigg\{ -\frac{1}{2} ( {\beta_{k}} - { \beta}^*_{ k})^\top ( A_{ k} +  U_{ k}/\tau) ( {\beta_{k}} - { \beta}^*_{ k} ) \bigg\} \prod_{i =1}^{| k|} \beta_{k_i}^{2r}  \\
	&&\qquad \times \,\, \exp \bigg\{ -\frac{1}{2} \what{ \beta}_{ k}^\top \big( A_{ k} -  A_{ k}( A_{ k} +  U_{ k}/\tau)^{-1}  A_{ k} \big)\what{ \beta}_{ k}   \bigg\}  \,d {\beta_{k}} d\tau\\
	&=&  \int(2\pi)^{| k|/2}  \det \big(A_{k} + U_{k}/\tau\big)^{-1/2}  \tau^{-r| k| - | k|/2 - \lambda_1 - 1}  \exp\Big(-\frac{\lambda_2}\tau\Big) E_{ k}\big(\prod_{i=1}^{| k|}\beta_{k_i}^{2r}\big)  \\
	&&\quad \times \,\, \exp\Big\{-\frac{1}{2} \what{ \beta}_{ k}^\top \big( A_{ k} -  A_{ k}( A_{ k} +  U_{ k}/\tau)^{-1}  A_{ k} \big)\what{ \beta}_{ k}   \Big\} \, d\tau,
	\eea
	where $E_{ k}(.)$ denotes the expectation with respect to a multivariate normal distribution with mean $ {\beta_{k}}^*$ and 
	covariance matrix $ A_{ k}+  U_{ k}/\tau$. It follows from Lemma 6 in the supplementary material for \cite{johnson2012bayesian} that
	\bea
	E_{ k}\big(\prod_{i=1}^{| k|}\beta_{k_i}^{2r}\big) 
	&\le& \left(\frac{n\Lambda_{| k|} + \tau^{-1}a_2 }{n\lambda +  \tau^{-1}a_1} \right)^{| k|/2}\left\{\frac{4V}{| k|} + \frac{4\left[(2r-1)!!\right]^{\frac 1 r}}{n(\lambda +  \tau^{-2})}\right\}^{r| k|}\\
	&\le& \left(\frac{n\Lambda_{| k|}}{n\lambda +  \tau^{-1}a_1} \right)^{| k|/2}2^{r| k| - 1}\Bigg\{\left(\frac{4V}{| k|}\right)^{r| k|} + \left(\frac{4\left[(2r-1)!!\right]^{\frac 1 r}}{n(\lambda +  \tau^{-1}a_1)}\right)^{r| k|}\Bigg\} \\
	&\le& \Big(\frac{\Lambda_k}\lambda\Big)^{|k|/2} \exp\Big(\frac{a_2k}{2n\Lambda_k\tau}\Big)\,2^{r| k| - 1}\Bigg\{\left(\frac{4V}{| k|}\right)^{r| k|} + \left(\frac{4\left[(2r-1)!!\right]^{\frac 1 r}}{n\lambda}\right)^{r| k|}\Bigg\},
	\eea
	where $V = \| {\beta_{k}}^*\|_2^2$ and the last inequalify follows from $1 + x \le \exp(x)$. Next, note that it follows from Lemma A.3 in the supplemental material for \citep{Naveen:2018} that %{\color{red}(*********$\| \beta_{0, k}\|_2^2 = O_p(\log p)$ $t$ is fixed ****************)}
	\bea
	V = \| {\beta_{k}}^*\|_2^2 \le \|\what{ \beta}_{ k}\|_2^2 &\le& \big(\|\what{ \beta}_{ k} -  \beta_{0, k}\|_2 + \| \beta_{0, k}\|_2\big)^2 \\
	&\le& \bigg(\sqrt{\frac{| k|\Lambda_{| k|}\log p}{n}} + \sqrt{\log p}\bigg)^2\\
	&\le& 2\bigg(\frac{| k|\Lambda_{| k|}\log p}{n} + \log p\bigg).
	\eea	
	Therefore,
	\bea
	\left(\frac V {| k|}\right)^{r| k|} \le \bigg(\frac{2\log p}{| k|} + o(1)\bigg)^{r| k|}.
	\eea
	Therefore, by noting that $\exp\big\{-1/2 \what{ \beta}_{ k}^\top \big( A_{ k} -  A_{ k}( A_{ k} +  U_{ k}/\tau)^{-1}  A_{ k} \big)\what{ \beta}_{ k}   \big\} \ge 1$, it follows from \eqref{marginal density} that, for some constant $C_1>0$,
	\bean \label{integral_B}
	\int \int_B M_1\, d {\beta_{k}} d\tau 
	&\le& (C_1\Lambda_k)^{| k|/2}  \Big(\frac{\log p}{| k|}\Big)^{r| k|} \nonumber\\
	&&\times \int \det \big(A_{k} + U_{k}/\tau\big)^{-1/2} \tau^{-r| k| - | k|/2 - \lambda_1 - 1}  \exp\bigg\{-\frac{\lambda_2 - a_2|k|/(2n\Lambda_k)}\tau\bigg\} \,d\tau \nonumber\\
	&\le& \Big(\frac{C_1\Lambda_k}{n}\Big)^{| k|/2}  \Big(\frac{\log p}{| k|}\Big)^{r| k|} \det(A_k)^{-1/2} \nonumber\\
	&&\times \int  \tau^{-r| k| - | k|/2 - \lambda_1 - 1}  \exp\bigg\{-\frac{\lambda_2 - a_2|k|/(2n\Lambda_k)}\tau\bigg\} \,d\tau \nonumber\\
	&\le& \Big(\frac{C_1\Lambda_k}{n}\Big)^{| k|/2}  \Big(\frac{\log p}{| k|}\Big)^{r| k|} \det(n^{-1}A_k)^{-1/2} \nonumber\\
	&&\times \big\{\lambda_2 - a_2|k|/(2n\Lambda_k)\big\}^{-(r|k| + |k|/2 + \lambda_1)} \Gamma\big(r|k| + |k|/2 + \lambda_1 \big).
	\eean
	Next, note that 
	\bea
	\int \int_{B^c} M_2 \, d {\beta_{k}} d\tau 
	&\le& \int \int \tau^{-r| k| - | k|/2 - \lambda_1 - 1} \exp\Big(-\frac{\lambda_2}\tau\Big) \exp \bigg(- \frac{ {\beta_{k}} ^\top  U_{ k}  {\beta_{k}}}{2 \tau} \bigg) \prod_{i =1}^{| k|} \beta_{k_i}^{2r} \, d {\beta_{k}} d\tau  \\
	&\le& \int \tau^{-r| k| - | k|/2 - \lambda_1 - 1} \exp\Big(-\frac{\lambda_2}\tau\Big) \big\{2\pi (2r-1)^2 \tau/a_1\big\}^{|k|/2} (\tau/a_1)^{r|k|} \, d\tau \\
	&\le& (C_2)^{|k|} \lambda_2^{-\lambda_1} \Gamma(\lambda_1).
	\eea
	Combining with \eqref{marginal density} and using the Stirling approximation for the gamma function, we obtain the following upper bound for $m_{k}(y_n)$,
	\bean \label{marginal_upper_supset}
		m_{ k}( y) &\le&  C_3^{|k|}  \exp \big\{ L_n(\what{ \beta}_{ k}) \big\} n^{-|k|/2} \det(n^{-1}A_k)^{-1/2} \lambda_2^{-(r|k| + |k|/2)}  \Lambda_k^{|k|/2} \big(|k|^{1/r} \log p \big)^{r|k|} \nonumber\\
		&&+ \,\, C_3^{|k|}  \exp \big\{ L_n(\what{ \beta}_{ k}) \big\} \exp \big( -\frac{1-\epsilon}{8} c^2 \lambda | k| \Lambda_{| k|} \log p  \big) \nonumber\\
		&\lesssim& C_3^{|k|}  \exp \big\{ L_n(\what{ \beta}_{ k}) \big\} n^{-|k|/2} \det(n^{-1}A_k)^{-1/2} \lambda_2^{-(r|k| + |k|/2)}  \Lambda_k^{|k|/2} \big(|k|^{1/r} \log p \big)^{r|k|},
	\eean
	for any $k \in M_1$ and some constant $C_3 > 0$, by noting that 	
	\bea
	\det (A_k)^{1/2} \lambda_2^{r|k| + |k|/2}  \Lambda_k^{-|k|/2} \big(|k|^{1/r} \log p \big)^{-r|k|} 
	&\ll&\exp \big\{ \frac{1-\epsilon}{8} c^2 \lambda | k| \Lambda_{| k|} \log p  \big\} 
	\eea
	by Conditions \hyperref[cond_A1]{\rm (A1)}, \hyperref[cond_A2]{\rm (A2)} and \hyperref[cond_A4]{\rm (A4)}. 
	
	Similarly, by Lemma 4 in the supplemental material for \cite{johnson2012bayesian} and the similar arguments leading up to \eqref{integral_B}, with probability tending to 1, we have, for some constant $C_4 > 0$, the marginal conditional density $m_t(y)$ will be bounded below by, %{\color{red}(*********$\| \beta_{0, k}\| \ge  \frac{\log p}{n}$, $t$ is fixed ****************)}
	\bea
	m_{ t}( y) &=& \int \int \exp \big\{L_n( {\beta_{t}}) \big\}  d_{ t}(2\pi)^{-| t|/2}\tau^{-r| t| - | t|/2} | U_{ t}| ^{\frac 1 2} \exp \Big(- \frac{ {\beta_{t}} ^\top  U_{ t}  {\beta_{t}}}{2 \tau}\Big)\prod_{i =1}^{| t|} \beta_{t_i}^{2r} \nonumber\\
	&&\qquad \times \,\,\frac{\lambda_2^{\lambda_1}}{\Gamma(\lambda_1)}\tau^{-\lambda_1 - 1} \exp\Big(-\frac{\lambda_2}\tau\Big)\, d {\beta_{t}} d\tau  \nonumber\\
	&\ge&  \frac{\lambda_2^{\lambda_1}}{\Gamma(\lambda_1)}d_{t}(2\pi)^{-|t|/2} | U_{t}| ^{\frac 1 2} \exp \big\{ L_n(\what{ \beta}_{t}) \big\}\nonumber\\
	&&\times \,\,\int \int \tau^{-r| t| - | t|/2 - \lambda_1 - 1}  \exp\Big(-\frac{\lambda_2}\tau\Big)\exp \bigg\{ -\frac{1}{2} ( {\beta_{t}} - { \beta}^*_{ t})^\top ( A^\prime_{ t} +  U_{ t}/\tau) ( {\beta_{t}} - { \beta}^*_{ t} ) \bigg\} \prod_{i =1}^{| t|} \beta_{t_i}^{2r}  \\
	&&\qquad \quad \times \,\, \exp \bigg\{ -\frac{1}{2} \what{ \beta}_{ t}^\top \big( A^\prime_{ t} -  A^\prime_{ t}( A^\prime_{ t} +  U_{ t}/\tau)^{-1}  A^\prime_{ t} \big)\what{ \beta}_{ t}   \bigg\}  \,d {\beta_{t}} d\tau\\
	&\ge&  \frac{\lambda_2^{\lambda_1}}{\Gamma(\lambda_1)}d_{t}(2\pi)^{-|t|/2} | U_{t}| ^{\frac 1 2} \exp \big\{ L_n(\what{ \beta}_{t}) \big\}(\log p)^{-2d - 2r|t|} \det(A^\prime_t)^{-1/2}\\
	&& \times \,\,  \int \tau^{-r|t| - |t|/2 - \lambda_1}  \exp\Big(-\frac{\lambda_2}\tau\Big) \exp\Big\{-\frac{|t|a_2}{2n(1+\epsilon)\lambda\tau}\Big\} \, d\tau\\
	&\ge&  C_4 \exp \big\{ L_n(\what{ \beta}_{t}) \big\} n^{-|t|/2} \det(n^{-1}A^\prime_t)^{-1/2} \lambda_2^{-(r|t| + |t|/2 - 1)} (\log p)^{-2d - 2r|t|},
	\eea
	by Lemma \ref{lem_aux1}, where $ A^\prime_{ t} = (1+\epsilon) H_n( \beta_{0, t})$.
	Therefore, with probability tending to 1, combining with \eqref{marginal_upper_supset}, for some constant $C^* > 0$,%{\color{red} $\tau \sim n^{-1/(2r+1)}p^{(1+\delta)}$ or $n^{-1/(2r+1)}p^{(2+\delta)}$}
	\bean \label{supset_upbound_ratio}
	\frac{m_{ k}( y)}{m_{ t}( y)}
	&\lesssim&  \big\{C^*n^{1/(2r+1)}\lambda\big\}^{-(r+ 1/2)(| k|-| t|)}  \exp \big\{ L_n(\what{ \beta}_{ k}) - L_n(\what{ \beta}_{ t})  \big\}\frac{\det (n^{-1}A^\prime_t)^{1/2} }{\det (n^{-1}A_k)^{1/2} } \nonumber\\
	&& \times \,\,  \Lambda_k^{|k|/2} \big(|k|^{1/r} \log p \big)^{r|k|} (\log p)^{2d + 2r|t|} \lambda_2^{-1} \nonumber\\
	&\lesssim& (C^*p)^{-(2+\delta/2) (|\bm k| - |\bm t|)} p^{ (1+\delta^*)(1+2w)(|\bm k| - |\bm t|)  },
	\eean
	for any $k \in M_1$, where the second inequality holds by Lemma 7.2 in \citep{Lee2019group} and by noting that it follows from Lemma 7.3 in \citep{Lee2019group},
	\bean\label{loglike_diff}
	L_n(\what{ \beta}_{ k}) - L_n(\what{ \beta}_{ t})  
	&\le& b_n (| k| - | t|)
	\eean
	for any $ k \in M_1$ with probability tending to 1, where $b_n = (1+\delta^*)(1+2w) \log p$ such that $1+\delta/2 > (1+\delta^*)(1+2w)$.
	
	Hence, with probability tending to 1, it follows from \eqref{supset_upbound_ratio} and \eqref{loglike_diff} that
	\begin{align} \label{supset_upbound_ratio_1}
	\sum_{ k\in M_1} PR( k,  t) = \sum_{ k \supset  t}\frac{\pi( k)m_{ k}( y)}{\pi( t)m_{ t}( y)}
	\le & \sum_{ k \supset  t}\frac{m_{ k}( y)}{m_{ t}( y)} \nonumber\\
	\lesssim & \sum_{| k| - | t| = 1}^{m_n - | t|} \binom {p-| t|}{| k|-| t|} (C^*p)^{-(1 + \delta^\prime) (| k| - | t|)},
	\end{align}
	for some constant $\delta^\prime > 0$. Using $\binom {p} {| k|} \le p^{| k|}$ and \eqref{supset_upbound_ratio_1}, we get
	\begin{align*}
	\sum_{ k\in M_1} PR( k,  t) \,\,=\,\, o(1).
	\end{align*}
	Thus, we have proved the desired result \eqref{sel_goal1}.
\end{proof}

\begin{proof}[Proof of Theorem \ref{thm:ratio}]
	Let $M_2 = \{ k :  k \nsupseteq  t,  | k|\le m_n \}$.
	For any $ k \in M_2$, let $ k^* =  k \cup  t$, so that $ k^* \in M_1$.
	Let $ \beta_{{ k^*}} $ be the $|{ k^*}|$-dimensional vector including $ {\beta_{k}}$ for $ k$ and zeros for ${ t \setminus  k}$.
	Then by Taylor's expansion and Lemmas A.1 and A.3 in \cite{Naveen:2018}, with probability tending to 1,
	\bea
	L_n ( \beta_{{ k^*}}) 
	&=& L_n ( \what{ \beta}_{{ k^*}} ) - \frac{1}{2} (  \beta_{{ k^*}} - \what{ \beta}_{{ k^*}} )^\top  H_n(\tilde{ \beta}_{{ k^*}})(  \beta_{{ k^*}} - \what{ \beta}_{{ k^*}} ) \\
	&\le& L_n ( \what{ \beta}_{{ k^*}} ) - \frac{1-\epsilon}{2} (  \beta_{{ k^*}} - \what{ \beta}_{{ k^*}} )^\top  H_n( { \beta}_{0,{ k^*}})(  \beta_{{ k^*}} - \what{ \beta}_{{ k^*}} )  \\
	&\le& L_n ( \what{ \beta}_{{ k^*}} ) - \frac{n (1-\epsilon)\lambda}{2} \| \beta_{{ k^*}} - \what{ \beta}_{{ k^*}}\|_2^2
	\eea
	for any $ \beta_{ k^*}$ such that $\| \beta_{{ k^*}} -  \beta_{0,{ k^*}} \|_2 \le c \sqrt{ | k^*| \Lambda_{| k^*|} \log p /n } = c w_n$ for some large constant $c>0$. Let $ B_{ k} = n(1-\epsilon)\lambda  I_{ k}$ and $ {\beta_{k}}^* = ( B_{ k} +  U_{ k}/\tau)^{-1}  B_{ k} \what{ \beta}_{ k}$. Define the set $B_* := \big\{  {\beta_{k}}:  \| \beta_{{ k^*}} - \what{ \beta}_{{ k^*}}\|_2 \le cw_n /2  \big\},$
	for some large constant $c>0$, then by similar arguments used for super sets, with probability tending to 1,
	\bea
	 \pi( k \mid  y)
	 &=& \int \int \exp \big\{L_n( {\beta_{k^*}}) \big\}  d_{ k}(2\pi)^{-| k|/2}\tau^{-r| k| - | k|/2} | U_{ k}| ^{\frac 1 2} \exp \Big(- \frac{ {\beta_{k}} ^\top  U_{ k}  {\beta_{k}}}{2 \tau}\Big)\prod_{i =1}^{| k|} \beta_{k_i}^{2r} \\
	 &\le&  \frac{\lambda_2^{\lambda_1}}{\Gamma(\lambda_1)}d_{ k}(2\pi)^{-| k|/2} | U_{ k}| ^{\frac 1 2} \exp \big\{ L_n(\what{ \beta}_{ k^*}) \big\}\\
	 &&\times\,\, \Bigg\{ \int \int_{B_*} N_1 \,d {\beta_{k}} d\tau  +\exp \big( -\frac{1-\epsilon}{8} c^2 \lambda | k^*| \Lambda_{| k^*|} \log p  \big) \int \int_{B_*^c} N_2\, d {\beta_{k}} d\tau \Bigg\},
	\eea
	where 
	\bea
	N_1 = \exp\Big(-\frac{\lambda_2}\tau\Big) \exp \bigg\{ - \frac{n(1-\epsilon)\lambda}{2} \|  \beta_{{ k^*}} - \what{ \beta}_{{ k^*}} \|_2^2 - \frac{ {\beta_{k}} ^\top  U_{ k}  {\beta_{k}}}{2 \tau}  \bigg\} \prod_{i =1}^{| k|} \beta_{k_i}^{2r} 
	\eea
	and 
	\bea
	M_2 = \tau^{-r| k| - | k|/2 - \lambda_1 - 1} \exp\Big(-\frac{\lambda_2}\tau\Big) \exp \bigg(- \frac{ {\beta_{k}} ^\top  U_{ k}  {\beta_{k}}}{2 \tau} \bigg) \prod_{i =1}^{| k|} \beta_{k_i}^{2r}.
	\eea
	Note that
	\bea
	&& \int \exp \Big\{  - \frac{n(1-\epsilon)\lambda}{2} \|  \beta_{{ k^*}} - \what{ \beta}_{{ k^*}} \|_2^2\Big\} \exp\Big(- \frac{ {\beta_{k}} ^\top  U_{ k}  {\beta_{k}}}{2 \tau}\Big)\prod_{i =1}^{| k|} \beta_{k_i}^{2r} d {\beta_{k}} \\
	&=& \int \exp \Big\{- \frac{n(1-\epsilon)\lambda}{2} \|  \beta_{{ k}} - \what{ \beta}_{{ k}} \|_2^2 - \frac{1}{2\tau}  {\beta_{k}} ^\top  U_{ k}  {\beta_{k}}  \Big\} \prod_{i =1}^{| k|} \beta_{k_i}^{2r} d {\beta_{k}} \, \exp \Big\{ - \frac{n(1-\epsilon)\lambda}{2} \|\what{ \beta}_{{ t\setminus  k}} \|_2^2  \Big\} \\
	&=& (2\pi)^{\frac{| k|}{2}} \det(B_{ k} +  U_{ k}/\tau)^{-1/2} \exp \left\{ -\frac{1}{2} \what{ \beta}_{ k}^\top \big( B_{ k} -  B_{ k}( B_{ k}+  U_{ k}/\tau)^{-1}  B_{ k} \big)\what{ \beta}_{ k}   \right\}E_{ k}\big(\prod_{i=1}^{| k|}\beta_{k_i}^{2r}\big)\\
	&&\times \,\,  \exp \Big\{ - \frac{n(1-\epsilon)\lambda}{2} \|\what{ \beta}_{ t\setminus  k} \|_2^2  \Big\}  .
	\eea
	where $E_k(.)$ denotes the expectation with respect to a multivariate normal distribution with mean $ {\beta_{k}}^*$ and 
	covariance matrix $ B_{ k}+  U_{ k}/\tau$. It follows from Lemma 6 in the supplementary material for \cite{johnson2012bayesian} that
	\bea
	E_{ k}\Big(\prod_{i=1}^{| k|}\beta_{k_i}^{2r}\Big) 
	&\le& \left(\frac{n\lambda + \tau^{-1}a_2 }{n\lambda +  \tau^{-1}a_1} \right)^{| k|/2}\left\{\frac{4V}{| k|} + \frac{4\left[(2r-1)!!\right]^{\frac 1 r}}{n(\lambda +  \tau^{-1}a_1)}\right\}^{r| k|}\\
	&\le& \exp\Big(\frac{a_2|k|}{2n\lambda\tau}\Big)2^{r| k| - 1}\left\{\left(\frac{4V}{| k|}\right)^{r| k|} + \left(\frac{4\left[(2r-1)!!\right]^{\frac 1 r}}{n\lambda}\right)^{r| k|}\right\},
	\eea
	where $V = \| {\beta_{k}}^*\|_2^2$. Therefore,
	for some constant $T_1>0$, we have
	\bean \label{integral_B*}
	\int \int_{B_*} N_1\, d {\beta_{k}} d\tau 
	&\le& \Big(T_1\frac{\log p}{| k|}\Big)^{r| k|} \exp \Big\{ - \frac{n(1-\epsilon)\lambda}{2} \|\what{ \beta}_{ t\setminus  k} \|_2^2  \Big\} \nonumber\\
	&&\times  \int \det \big(B_{k} + U_{k}/\tau\big)^{-1/2} \tau^{-r| k| - | k|/2 - \lambda_1 - 1}  \exp\bigg\{-\frac{\lambda_2 - a_2|k|/(2n\lambda)}\tau\bigg\} \,d\tau \nonumber\\
	&\le& T_1^{|k|}n^{-| k|/2}  \Big(\frac{\log p}{| k|}\Big)^{r| k|} \det(n^{-1}B_k)^{-1/2} \exp \Big\{ - \frac{n(1-\epsilon)\lambda}{2} \|\what{ \beta}_{ t\setminus  k} \|_2^2  \Big\} \nonumber\\
	&&\times \big\{\lambda_2 - a_2|k|/(2n\lambda)\big\}^{-(r|k| + |k|/2 + \lambda_1)} \Gamma\big(r|k| + |k|/2 + \lambda_1 \big).
	\eean
	Using similar arguments leading up to \eqref{marginal_upper_supset}, since the lower bound for $\pi( t\mid  y)$ can be derived as before, we obtain the following upper bound for the posterior ratio, for some  constants $T_2 > 0$ and $c > 0$,
	\bean
	 PR( k,  t) 
	&\lesssim& 
	\big\{T_2n^{1/(2r+1)}\lambda_2\big\}^{-(r + 1/2)(| k|-| t|)} 
	\frac{ \det \big\{ (1+\epsilon)n^{-1}  H_n( \beta_{0, t}) \big\}^{1/2} }{\det \big\{(1-\epsilon)\lambda  I_{ k} \big\}^{1/2}} \nonumber \\
	&&\times \,\, \exp \big\{ L_n(\what{ \beta}_{{ k^*}}) - L_n(\what{ \beta}_{ t})  \big\}  \,  \exp \Big\{ - \frac{n(1-\epsilon)\lambda}{2} \|\what{ \beta}_{{ t\setminus  k}} \|_2^2  \Big\} \label{M2_part1}\\
	&+& \big\{T_2n^{1/(2r+1)}\lambda_2\big\}^{-(r+1/2)(| k|-| t|)}  { \det \big\{ (1+\epsilon)n^{-1}  H_n( \beta_{0, t}) + (n\tau)^{-1} U_{ t}  \big\}^{1/2} } \nonumber \\
	&&\times \,\, \exp \big\{ L_n(\what{ \beta}_{{ k^*}}) - L_n(\what{ \beta}_{ t})  \big\} \exp \big\{  - c\, C | k^*| \Lambda_{| k^*|} \log p \big\}  \label{M2_part2}
	\eean
	for any $ k \in M_2$ with probability tending to 1.
	
	We first focus on \eqref{M2_part1}.
	Note that
	\bea
	\frac{ \det \big\{ (1+\epsilon)n^{-1}  H_n( \beta_{0, t}) \big\}^{1/2} }{\det \big\{(1-\epsilon)\lambda  I_{ k} \big\}^{1/2}} 
	&\le& \frac{  \big\{ (1+\epsilon) \Lambda_{| t|} \big\}^{| t|/2} }{ \big\{(1-\epsilon)\lambda \big\}^{| k|/2} } \\
	&\lesssim& \exp \big\{ T_3| t| \log  \Lambda_{| t|}   \big\} \Big\{ \frac{1}{(1-\epsilon)\lambda } \Big\}^{(|k|-|t|)/2}
	\eea
	for some constant $T_3>0$.
	Furthermore, by the same arguments used in \eqref{loglike_diff}, we have
	\bea
	L_n(\what{ \beta}_{{ k^*}}) - L_n(\what{ \beta}_{ t})  
	&\lesssim& C_*(|{ k^*}| - | t|)\log p  \\
	&=& C_*|{ t\setminus  k}|\log p +  C_*( | k|- | t|)\log p
	\eea
	for some constant $C_* >0$ and for any $ k \in M_2$ with probability tending to 1.
	Here we choose $C_* = (1+\delta^*)(1+2w)$ if $| k|> | t|$ or $C_* = 3+\delta$ if $| k| \le | t|$ so that
	\bea
	&&(T_2 p)^{-(1+\delta)(| k|-| t|)} \Big\{ \frac{1}{(1-\epsilon)\lambda } \Big\}^{(|k|-|t|)/2} p^{ C_* ( | k|- | t| )} \\
	&\lesssim& p^{ (C_* - 2 - \delta) ( | k|- | t| ) } \,\,=\,\,o(1),
	\eea
	where the inequality holds by Condition \hyperref[cond_A4]{\rm (A4)}.
	To be more specific, we divide $M_2$ into two disjoint sets $M_2' = \{ k :  k \in M_2 ,  | t| < | k| \le m_n \}$ and $M_2^* = \{ k :  k \in M_2 , | k| \le | t| \}$, and will show that $\sum_{ k \in M_2'} PR( k,  t)  + \sum_{ k \in M_2^*} PR( k,  t)\lra 0$ as $n\to\infty$ with probability tending to 1.
	Thus, we can choose different $C_*$ for $M_2'$ and $M_2^*$ as long as $C_* \ge (1+\delta^*)(1+2w)$.	
	On the other hand, with probability tending to 1, by Condition \hyperref[cond_A3]{\rm (A3)},
	\bea
	\exp \Big\{ - \frac{n(1-\epsilon)\lambda}{2} \|\what{ \beta}_{{ t\setminus  k}} \|_2^2  \Big\} 
	&\le&  \exp \Big[ - \frac{n(1-\epsilon)\lambda}{2}  \Big\{  \|  \beta_{0,{ t\setminus  k}} \|_2^2 - \|\what{ \beta}_{{ t\setminus  k}} -  \beta_{0,{ t\setminus  k}} \|_2^2 \Big\}  \Big]  \\
	&\le&  \exp \Big[ - \frac{n(1-\epsilon)\lambda}{2}  \Big\{  | t\setminus  k|^2 \min_{j \in  t}\|\beta_{0,j}\|_2^2 - c' w_n'^2 \Big\}  \Big]  \\
	&\le& \exp \Big\{ - \frac{(1-\epsilon)\lambda}{2}   \big( c_0 | t\setminus  k|^2 | t|    - c' |{ t\setminus  k}|  \big)  \Lambda_{| t|} \log p   \Big\}  \\
%	&\le& \exp \Big\{ - \frac{(1-\epsilon)\lambda}{3}   c_0 | t\setminus  k|^2 | t|  \Lambda_{| t|} \log p   \Big\}  
	&\le& \exp \Big\{  - \frac{(1-\epsilon)\lambda}{2}  ( c_0 - c') | t\setminus  k|^2 | t|  \Lambda_{| t|} \log p    \Big\}
	\eea
	for any $k \in M_2$ and some large constants $c_0> c'>0$, where $w_n'^2 = |{ t\setminus  k}| \Lambda_{| t\setminus  k|} \log p/n $.
	Here, $c'  = 5 \lambda^{-2} (1-\epsilon)^{-2}$ by the proof of Lemma A.3 in \cite{Naveen:2018}.
	%	assume $(\beta_{0,t})_i^2 \ge c_0 | t| \Lambda_{| t|} \log p /n$ for some large $c>0$
	
	Hence,  \eqref{M2_part1} for any $k \in M_2$ is bounded above by
	\bea
	%	\sum_{k\in M_2} PR(k,t) &\lesssim& 
&& (T_2p)^{-(|k| - |t|)}  \exp\Big\{  | t| \log \Lambda_{| t|} + C_* |{ t\setminus  k}| \log p - \frac{(1-\epsilon)\lambda}{2}   (c_0- c') | t\setminus  k|^2 | t|  \Lambda_{| t|} \log p  \Big\}   \\
	&\lesssim&    \exp \Big\{ - \Big( \frac{(1-\epsilon)\lambda}{2}  (c_0- c')  - C_* - o(1) \Big) | t\setminus  k|^2 | t|  \Lambda_{| t|} \log p \Big\} \\
	&\le&    \exp \Big\{ - \Big( \frac{(1-\epsilon)\lambda}{2}  (c_0- c')  - C_*  - o(1) \Big) | t\setminus  k|^2 | t|  \Lambda_{| t|} \log p  \Big\} \\
	&\le&  \exp \Big\{ - \Big( \frac{(1-\epsilon)\lambda}{2}  (c_0- c')  - C_* - o(1) \Big)  | t|  \Lambda_{| t|} \log p  \Big\} 
	\eea
	with probability tending to 1, where the last term is of order $o(1)$ because we assume $c_0 = \frac{1}{(1-\epsilon_0)\lambda}\big\{2(3+\delta) + \frac{5}{(1-\epsilon_0)\lambda} \big\} > \frac{2}{(1-\epsilon)\lambda}(C_* + o(1)) + c'$.
	
	It is easy to see that the maximum \eqref{M2_part2} over $k \in M_2$ is also of order $o(1)$ with probability tending to 1 by  the similar arguments.
	Since we have \eqref{sel_goal1} in the proof of Theorem \ref{thm:nosuper}, it completes the proof.
\end{proof}

\begin{proof}[Proof of Theorem \ref{thm:selection}]
	Let $M_2 = \{ k :  k \nsupseteq  t,  | k|\le m_n \}$.
	Since we have Theorem \ref{thm:nosuper}, it suffices to show that 
	\bean\label{sel_cons_goal}
	\sum_{ k:  k \in M_2} PR( k,  t) &\overset{P}{\lra}& 0 \quad \text{ as } n\to\infty.
	\eean	
	By the proof of Theorem \ref{thm:ratio}, the summation of \eqref{M2_part1} over $k \in M_2$ is bounded above by
	\bea
%	\sum_{k\in M_2} PR(k,t) &\lesssim& 
	&&\sum_{ k\in M_2}(T_2\lambda_2)^{-r(|k| - |t|)} \exp \Big\{ -\Big( \frac{(1-\epsilon)\lambda}{2}  (c_0- c')  - C_* - o(1) \Big) | t\setminus  k|^2 | t|  \Lambda_{| t|} \log p    \Big\}   \\
	&\le& \sum_{| k|=0}^r \sum_{v = 0}^{(| t|-1) \wedge | k|} \binom{| t|}{v} \binom{p-| t|}{| k|-v} p^{-( | k|- | t| ) } \exp \Big\{ - \Big( \frac{(1-\epsilon)\lambda}{2}  (c_0- c')  - C_* - o(1) \Big)(| t|-v)^2 | t|  \Lambda_{| t|} \log p    \Big\} \\
	&\le& \sum_{| k|=0}^p \sum_{v = 0}^{(| t|-1) \wedge | k|}  \big( | t| p \big)^{| t| -v} \exp \Big\{ - \Big( \frac{(1-\epsilon)\lambda}{2}  (c_0- c')  - C_* - o(1) \Big) (| t|-v)^2 | t|  \Lambda_{| t|} \log p    \Big\} \\
	&\le& \exp \Big\{ - \Big( \frac{(1-\epsilon)\lambda}{2}  (c_0- c')  - C_* - o(1) \Big)  | t|  \Lambda_{| t|} \log p + 2(| t|+2) \log p    \Big\} \\
	&\le& \exp \Big\{ -   \Big( \frac{(1-\epsilon)\lambda}{2}  (c_0- c')  - C_* - 6 - o(1) \Big)  | t|  \Lambda_{| t|} \log p \Big\}
	\eea
	with probability tending to 1, where $C_* \le 3 +\delta$ is defined in the proof of Theorem \ref{thm:ratio}.
	Note that the last term is of order $o(1)$ because we assume $c_0 = \frac{1}{(1-\epsilon_0)\lambda}\big\{2(9 + 2\delta) + \frac{5}{(1-\epsilon_0)\lambda} \big\} > \frac{2}{(1-\epsilon)\lambda}(C_* + 6 + o(1)) + c'$.
	It is easy to see that the summation of \eqref{M2_part2} over $k \in M_2$ is also of order $o(1)$ with probability tending to 1 by  the similar arguments.
\end{proof}

\begin{lemma}\label{lem_aux1}
	Under Condition \hyperref[cond_A2]{\rm (A2)}, we have
	\bea
	\exp \big\{ \frac{1}{2} \what{ \beta}_{ t}^\top \big(  A^\prime_{t} -  A^\prime_{t}( A^\prime_{t}+\tau^{-1} U_{ t})^{-1}  A^\prime_{t} \big)\what{ \beta}_{ t}   \big\}  
	&\lesssim& \tau^{-1}(\log p)^{2d}
	\eea
	with probability tending to 1.	
\end{lemma}

\begin{proof}
	Note that, by Condition \hyperref[cond_A2]{\rm (A2)},
	\bea
	( A^\prime_{t} + \tau^{-1}  U_{ t})^{-1} &\ge& ( A^\prime_{t} + a_2(n\tau \lambda)^{-1}  A^\prime_{t})^{-1} \,\,=\,\, \frac{n\tau \lambda}{n\tau \lambda+a_2}  {A^\prime_{t}}^{-1},
	\eea
	which implies that
	\bea
	\frac{1}{2} \what{ \beta}_{ t}^\top \big(  A^\prime_{t} -  A^\prime_{t}( A^\prime_{t}+\tau^{-1 } U_{ t})^{-1}  A^\prime_{t} \big)\what{ \beta}_{ t}  
	&\le& \frac{a_2}{2(n\tau \lambda+a_2)} \what{ \beta}_{ t}^\top  A^\prime_{t} \what{ \beta}_{ t}   .
	\eea
	Thus, we complete the proof if we show that
	\bea
	 \frac{1}{ n\tau } \what{ \beta}_{ t}^\top  H_n( \beta_{0, t}) \what{ \beta}_{ t} 
	 &\lesssim& \tau^{-1}(\log p)^{2d}
	\eea
	for any $t\in M_1$ with probability tending to 1.
	By Lemma A.3 in \cite{Naveen:2018} and Condition \hyperref[cond_A2]{\rm (A2)},
	\bea
	 \frac{1}{n \tau}\what{ \beta}_{ t}^\top  H_n( \beta_{0, t}) \what{ \beta}_{ t}
	 &\le& \frac{1}{\tau} \lambda_{\max}\big\{ n^{-1}  H_n( \beta_{0, t}) \big\} \|\what{ \beta}_{ t}\|_2^2 \\
	 &\le& \frac{1}{\tau} (\log p)^d \big( \| \beta_{0, t}\|_2^2  + o(1) \big) 
	 \,\,=\,\, O\big(\tau^{-1}(\log p)^{2d}\big)
	\eea
	with probability tending to 1.
\end{proof}

\bibliographystyle{plainnat} 
\bibliography{references}
\end{document}